\documentclass[11pt,a4paper]{article}

\usepackage[utf8]{inputenc} 
\usepackage[T1]{fontenc}    %
\usepackage[english]{babel} 
\usepackage[final]{microtype} 

\usepackage{amssymb,amsmath,mathtools} 
\usepackage{amsthm} 

\usepackage{xcolor}
\usepackage{bbm}
\usepackage{braket}

\usepackage{pgf}
\usepackage{pgfplots}
\usepackage{tikz}
\usetikzlibrary{arrows,calc}
\usepackage{verbatim}
\usetikzlibrary{decorations.pathreplacing,decorations.pathmorphing}
\usepackage[numbers,sort&compress]{natbib}
\usepackage{dsfont}

\usepackage[hmargin=0.12\paperwidth,vmargin=0.16\paperwidth,bindingoffset=0cm]%
{geometry} 

\numberwithin{equation}{section} 

\theoremstyle{plain}
  
  \newtheorem{prop}{Proposition}
  \newtheorem{lemma}{Lemma}
  
\theoremstyle{definition}
  
  \newtheorem{remark}{Remark}
  
  \numberwithin{prop}{section}
   \numberwithin{cor}{section}
   \numberwithin{remark}{section}

\title{\Large\bfseries Moments of characteristic polynomials for classical $\beta$ ensembles}%
   
\author{Bo-Jian Shen and Peter J. Forrester}
\date{}

\begin{document}

\maketitle

School of Mathematics and Statistics,  The University of Melbourne,
Victoria 3010, Australia. \: \: Email: {\tt pjforr@unimelb.edu.au};  {\tt bojian.shen@unimelb.edu.au} \\

\bigskip

\begin{abstract}
\noindent
For random matrix ensembles with unitary symmetry, there is interest
in the large $N$ form of the moments of the absolute value of the characteristic polynomial for their relevance to the Riemann zeta function on the critical line , and to Fisher-Hartwig asymptotics in the theory of Toeplitz determinants. The constant (with respect to $N$) in this asymptotic expansion, involving the Barnes $G$ function, is most relevant to the first of these, while the algebraic term (in $N$) and the functional dependence on the power are of primary interest in the latter. Desrosiers and Liu \cite{DL14} have obtained the analogous expansions for the classical Gaussian, Laguerre and Jacobi $\beta$ ensembles in the case of even moments. We give simplified working of these results --- which requires the use of duality formulas and the use of steepest descents for multidimensional integrals --- providing too an error bound on the resulting asymptotic expressions. The universality of the constant term with respect to an earlier result known for the circular $\beta$ ensemble is established, which requires writing it in a Barnes $G$ function form, while the functional dependence on the powers is related to that appearing in Gaussian fluctuation formulas for linear statistics. In the Laguerre and Jacobi cases our working can be extended to the circumstance when the exponents in the weight function are (strictly) proportional to $N$, giving results not previously available in the literature.
\end{abstract}

\vspace{3em}

\section{Introduction}
\subsection{Settings for moments of characteristic polynomials}
In random matrix theory, a $\beta$ ensemble for Hermitian matrices (all eigenvalues real) refers to the point process
specified by a probability density function (PDF) with  functional form proportional to
\begin{equation}\label{1a}
\prod_{l=1}^N w(x_l) \prod_{1 \le j < k \le N} |x_k - x_j|^\beta.
\end{equation}
The one-body factor $w(x)$ in the language of statistical mechanics, or weight function as termed in orthogonal polynomial theory, 
thus distinguishes different $\beta$ ensembles, which may be denoted ME${}_{\beta,N}[w]$.
Of particular prominence are the choices (up to possible scaling of $x$)
\begin{equation}\label{1b}
w(x) = \begin{cases} e^{-x^2}, & {\rm Gaussian} \\
x^a e^{-x} \mathbbm 1_{x > 0} ,& {\rm Laguerre} \\
x^{a_1} (1 - x)^{a_2}  \mathbbm 1_{0 < x < 1},  & {\rm Jacobi}. \end{cases},
\end{equation}
Here the terminologies for the weights are those coming from orthogonal polynomial theory. For $\beta = 1,2$ and 4 classical results in random
matrix theory give constructions in terms of matrices with independent standard Gaussian entries (real entries for $\beta = 1$,
complex entries for $\beta = 2$ and quaternion entries for $\beta = 4$ realised as particular $2 \times 2$ complex blocks) which realise the eigenvalue
PDF (\ref{1a}) with weights (\ref{1b}). For example, with $G_{n,N}$ denoting an $n \times N$ rectangular standard complex Gaussian matrix
(also known as a rectangular complex Ginibre matrix \cite{BF24}), the construction ${1 \over 2} (G_{N,N} + G_{N,N} ^\dagger)$ specifies
the Gaussian unitary ensemble  which has eigenvalues PDF given by (\ref{1a}) with $\beta = 2$ and Gaussian weight, while the construction
$G_{n,N}^\dagger G_{n,N}$ with $n \ge N$ specifies the Laguerre unitary ensemble (also known as complex Wishart matrices) 
which has eigenvalues PDF given by (\ref{1a}) with $\beta = 2$ and Laguerre weight for $a = n - N$. The product of one body factors $\prod_{l=1}^N w(x_l)$
is identified with the PDF on the matrix space, while the product over pairs $ \prod_{1 \le j < k \le N} |x_k - x_j|^\beta$ is the eigenvalue dependent factor of the
Jacobian in the change of variables from the eigenvalues to the eigenvectors.
More on these constructions
 is contained in \cite[Ch.~1 and 3]{Fo10}. In fact underlying random matrix ensembles are known for all $\beta > 0$, and continuous
Laguerre parameter $a$ and Jacobi parameters $(a_1, a_2)$, which in the Gaussian and Laguerre cases involves certain tridiagonal
matrices \cite{DE02}.

Fundamental in matrix theory the characteristic polynomial. With $\{ x_j \}_{j=1}^N$ the eigenvalues, this corresponds to the product
$\prod_{j=1}^N ( x - x_j)$, which for a random matrix is a random function. From multiple viewpoints, there is interest in moments
of this random function with respect to the PDF for $\{ x_j \}_{j=1}^N$. Perhaps the simplest result of this kind for a $\beta$ ensemble
is for the average of the characteristic polynomial in the case of a classical weight. As an explicit example, in the case of the Gaussian
$\beta$ ensemble one has \cite{BF97a,De09}
 \begin{equation}\label{R1}
  \Big  \langle   \prod_{j=1}^N (x - x_j ) \Big \rangle_{{\rm ME}_{\beta,N}[e^{-\beta \lambda^2/2}]} = 2^{-N} H_N(x),
  \end{equation}  
  where $H_N(x)$ denotes the Hermite polynomial. Underlying (\ref{R1}) is the so-called duality identity
\cite{BF97a}, \cite{De09}, \cite[Eq.~(13.162)]{Fo10}
\begin{equation}\label{R2}
    \Big  \langle    \prod_{j=1}^N ( x  -  \sqrt{\alpha} x_j )^p     \Big  \rangle_{ {\rm ME}_{2/\alpha,N}[e^{-\lambda^2}]} =
   i^{-pN} \Big  \langle   \prod_{j=1}^p  ( i  x  -  x_j)^N \Big  \rangle_{{\rm ME}_{2 \alpha,p}[e^{-\lambda^2}]}
   \end{equation}  
valid for general $\alpha > 0$; the recent work \cite{Fo25+} reviews the broader context. The result (\ref{R1}) can be obtained from (\ref{R2}) by taking
$p=1$ and using a standard one-dimensional integral formula for the Hermite polynomials on the RHS.
With $p=2/\alpha$ and even, the LHS of (\ref{R2} relates to the density of the Gaussian $\beta$ ensemble.
This is true for the general class of PDFs (\ref{1a}), as follows 
from the fact that the density  is by definition proportional to the integral over
all but one coordinate (to begin, a system of $N+1$ eigenvalues should then be taken). The significance of the
duality (\ref{R2}) is that the RHS is then well suited to asymptotic analysis, both in the bulk and at the spectrum
edge \cite{BF97a,DF06,DL14,DL15}.

Moments of characteristic polynomials of random matrices also show up as the normalisation of a random function
formed out of the power of the characteristic polynomial of a random matrix,
\begin{equation}\label{R3}
{|\det (x \mathbb I_N - X) |^\alpha \over \langle |\det (x \mathbb I_N - X) |^\alpha  \rangle_X}
  \end{equation}  
  for several prominent ensembles. One of these is the circular $\beta$ ensemble of random unitary matrices
  (CE${}_{\beta,N}$ say),
  specified by the eigenvalue PDF proportional to
 \begin{equation}\label{R3a} 
 \prod_{1 \le j < k \le N} | e^{i \theta_k} - e^{i \theta_j} |^\beta.
   \end{equation}  
   For $\beta = 2$, (\ref{R3a}) is the eigenvalue PDF for matrices from the classical group $U(N)$ chosen with Haar
   measure; see \cite{DF17} for historical remarks. From theory relating to the average traces of this ensemble
   \cite{DS94,HKO01}, it is straightforward to show formally at least that for large $N$ the random function
   $\log \det (\mathbb I_N - z U^\dagger)$ limits to the random power series $G(z) = \sum_{j=1}^\infty {z^j \over \sqrt{j}} \mathcal N_j$, with
   each $\mathcal N_j$ an independent standard complex random variable. Moreover, this random power series has for
   its covariance $\langle \overline{G(z)} G(w) \rangle = - \log (1 - \bar{z} w)$. One significance of this is, again formally
   at least, that such a log-correlated field gives a realisation of Gaussian multiplicative chaos \cite{Ka85},
   i.e.~a random measure $\mu^{(\alpha)}$ such that
    \begin{equation}\label{R3b} 
    {d \mu^{(\alpha)} \over d \mu}(z) = {e^{\alpha G(z)} \over \langle e^{\alpha G(z)} \rangle},
     \end{equation}  
     where $\mu$ is the uniform measure, and it is required that $\alpha < 1/2 $ for the measure to be non-degenerate.
     Specifically, it was  proved in \cite{We15} (see too \cite{NSW20,BF22}) that for Haar distributed matrices from $U(N)$ there is convergence,
     for $N \to \infty$, of (\ref{R3}) to (\ref{R3b}) for $\alpha>0$ and small enough. In \cite{LN24} the range of $\alpha$ values is extended
     to allow for some negative values, and convergence is established too in relation to (\ref{R3}) for the the PDF (\ref{R3a}) with
     general $\beta > 0$. Beyond the circular ensembles, convergence of (\ref{R3}) to (\ref{R3b}) has been established for
     the Gaussian unitary ensemble \cite{BWW18}, which was extended to the Gaussian orthogonal and symplectic ensembles
     in \cite{Ki24}; the case of random real orthogonal and unitary symplectic matrices is considered in \cite{FK21}.

   For the circular $\beta$ ensemble, one has the average power of the characteristic polynomial
\begin{equation}\label{R4}   
\Big \langle \prod_{l=1}^N |e^{i \phi}  - e^{i \theta_l} |^\gamma  \Big \rangle_{{\rm CE}_{\beta,N}}, \qquad - \pi < \phi \le \pi,
  \end{equation}  
  is independent of $\phi$. Its value is known as a product of gamma functions from the theory of the Selberg integral;
  see \cite[Eq.~(4.4)]{Fo10}. In the log-gas picture of (\ref{R3a}) \cite{Dy62}, the average (\ref{R4}) with $\gamma = \beta p$ has an interpretation in
  terms of the dimensionless chemical potential $\beta \mu_N^*$ due to the introduction of an impurity charge of strength $p$ at the origin,
 \begin{equation}\label{R4a} 
 e^{-   \beta \mu_N^*} = R^{-p^2 \beta / 2} \Big \langle   \prod_{l=1}^N |1 - e^{i \theta_l}|^{\beta p}
 \Big  \rangle_{{\rm CE}_{\beta,N}},
   \end{equation}  
   where $R = N/(2 \pi)$ is a scale factor with the interpretation of the radius of the dilated unit circle  as required for the average eigenvalue spacing to
   be unity. In the case of $\beta/2$ a rational number, written in the form $\beta /2 = m/n$ for $m,n$ relatively prime, it was shown in 
   \cite{Fo92j} based on the exact evaluation of the circular $\beta$ ensemble average as a product of gamma functions using Selberg integral theory (see also \cite[Exercises 4.8 q.4]{Fo10}) that
   \begin{equation}\label{3.qb1}
(2 \pi )^{-p^2 \beta} \lim_{N \to \infty}e^{-\beta \mu^*_N} \Big |_{\beta p \mapsto 2p}
=  \tilde{A}_{\beta,p} |_{\beta/2=m/n},
\end{equation}
where
\begin{equation}\label{3.qb1+}
\tilde{A}_{\beta,p} |_{\beta/2=m/n} := n^{-2 p^2/\beta}
\prod_{\nu=0}^{n-1} \prod_{\mu=0}^{m-1}
{ G^2(p/m+\nu/n-\mu/m+1) \over
G(2p/m+\nu/n-\mu/m+1) G(\nu/n-\mu/m+1) },
\end{equation}
valid for general Re$(p) > - {1 \over 2}$.
Here $G(z)$ is the Barnes $G$ function, which satisfies the recurrence involving the gamma function,
$G(z+1) = \Gamma(z) G(z)$.

Important for the present viewpoint is that in addition to exhibiting a well defined and explicit bulk scaling limit, the result (\ref{3.qb1}) can be recast as an asymptotic formula for the average in (\ref{R4a}) itself. Thus for large $N$
\begin{equation}\label{cc1}
\Big \langle   \prod_{l=1}^N |1 - e^{i \theta_l}|^{\beta p}
 \Big  \rangle_{{\rm CE}_{\beta,N}}
 = (N/2\pi)^{p^2 \beta/2}
 \tilde{A}_{\beta,p} |_{\beta/2=m/n} \Big ( 1 + {\rm o}(1) \Big ).
\end{equation}
In contrast to the bulk scaling limit in (\ref{3.qb1+}), the result (\ref{cc1}) relates to a global scaling limit, in which the domain --- the unit circle --- of finite length is kept constant while the number of eigenvalues tends to infinity. Moreover, 
in \cite{BF97f}, upon the rewrite $  \prod_{l=1}^N |1 - e^{i \theta_l}|^{\beta p} = e^{\beta p \sum_{l=1}^N \log |1 - e^{i \theta_l}|}$, the average
in (\ref{cc1}) has been given the interpretation of the characteristic function of the linear statistic $ \sum_{l=1}^N \log |1 - e^{i \theta_l}|$; see
also the recent work \cite{BD23}. 

A particularly novel consequence of (\ref{3.qb1}) in the case $\beta = 2$ (and thus $n=m=1$) and $p$ a positive integer,
when the RHS simplifies to
 \begin{equation}\label{R4b} 
 { G^2(p+1) \over
G(2p+1) },
\end{equation}
was found by Keating and Snaith \cite{KS00a} in their now celebrated work linking the characteristic polynomial of Haar unitary matrices
and the Riemann zeta function on the critical line. Specifically, the random matrix average (\ref{R4}) with $\beta = 2$ was introduced in the context
of studying the large $T$ asymptotics of the zeta function moments ${1 \over T} \int_0^T | \zeta(1/2+it)|^{2p} \, dt$. The latter are conjectured to
have the leading form $g_p (\log T)^{p^2} \alpha_p$, where $\alpha_p$ is a known number theoretic quantity involving primes, for some
$g_q$. This has been proved in the cases $p=1$ and $p=2$, where it was furthermore shown that
$g_1 = 1$ \cite{HL18}, $g_2=2$ \cite{In26}. Upon some assumptions, this form has been verified in the cases $p=3$ and $p=4$ by
Conrey and collaborators, with the values $g_3 = 42$ and $g_4 = 24024$; see \cite{CG01}. The hypothesis of Keating and Snaith identifies
the terms $g_q  (\log T)^{p^2} $ with the leading large $N$ form of (\ref{R4}) with $\beta = 2$ upon the identification $N = \log T$, and this
predicts the value (\ref{R4b}) for $g_p$. This predicted value is indeed consistent with the values of $g_1,\dots,g_4$ as listed above, derived
from  calculations in analytic number theory.

Important for motivation of the present paper is the work of Br\'ezin and Hikami \cite{BH00} on an alternative random matrices derivation of
the $g_p$. This comes about through considering not an average of the power of the characteristic polynomial
over Haar distributed matrices from $U(N)$, but rather over matrices from the
(scaled) Gaussian unitary ensemble. For this it was shown that 
for large $N$ and $p$ a positive integer, and for $\lambda$ such that $|\lambda| < 1$,
  \begin{equation}\label{6.1B}   
\Big \langle \prod_{l=1}^{N} |\lambda - x_l|^{2p} \Big \rangle_{{\rm ME}_{2,N}[e^{-2 N x^2}]} =
( \pi N \rho^{\rm W}(\lambda))^{p^2} e^{2Np( \lambda^2 - 1/2-\log \, 2)} \prod_{l=0}^{p-1} {l! \over (p + l)!} \Big ( 1 + {\rm O} \Big ( {1 \over N}  \Big )\Big ),
   \end{equation} 
   where $\rho^{\rm W}(\lambda)$ is density function for the Wigner semi-circle law, $\rho^{\rm W}(\lambda) = {2 \over \pi} (1 - \lambda^2)^{1/2} \mathbbm 1_{|\lambda|<1}$.
Here the RHS has the functional form 
\begin{equation}\label{S1a}
(N \phi_2(\lambda))^{p^2} e^{Np \phi_1(\lambda)} g_p,
\end{equation} 
for particular $\phi_1, \phi_2$, with the identification of $g_p$ as specified by (\ref{R4b})
following from the identity
 $ \prod_{l=0}^{p-1} {l! \over (p + l)!}  = (G(1+p))^2/G(1+2p)$. Hence there is a form of universality with respect to (\ref{R4b}) from the random matrix viewpoint,
 as it appears systematically in both the circular and Gaussian $\beta = 2$ computation of the asymptotic form of the averaged power of the characteristic polynomial, notwithstanding
that the latter depends on $\lambda$. 

Subsequent to  \cite{BH00}, 
 several works have rederived and extended (\ref{6.1B}), 
motivated first by its interpretation as a Hankel determinant with a singularity in the corresponding generating function symbol of Fisher-Hartwig type
(the circular ensemble average (\ref{R4}) for $\beta = 2$ relates to Toeplitz determinants of this type; see e.g.~the review \cite{Kr11})
\cite[Eq.~(21) with $m=1$, after minor correction]{Ga05},
 \cite[Th.~1 with $m=1$]{Kr07}, \cite{Ch19},
and later by
its relevance to multiplicative Gaussian chaos as noted in the paragraph including (\ref{R3}); see \cite{Fo23t} and references therein.
Moreover, the analogue of (\ref{6.1B}) (and generalisations) have been considered for the Laguerre and Jacobi unitary
 ensembles \cite{CG21}. Inspection of the results of \cite{CG21} in this setting, which moreover are proved  for general Re$\, (p) > - {1 \over 2}$, reveal the structure noted in the sentence below (\ref{6.1B}) and thus in particular with (\ref{R4b}) as a distinguished factor. 

 A generalisation of (\ref{6.1B}) which occurs in applications \cite{FFGW03} has the
 power of a single characteristic polynomial replaced by $R$ such terms being averaged;
 thus $\prod_{l=1}^{N} |\lambda - x_l|^{2p}$ is replaced by $\prod_{i=1}^R\prod_{l=1}^{N} |\lambda_i - x_l|^{2p_s}$. The corresponding large $N$ asymptotics was conjectured in
 \cite{FF04} to contain as factors the RHS of (\ref{6.1B}) with $p \mapsto p_i$ for each
 $i=1,\dots,R$, a structure which has subsequently been rigorously established in
 \cite{Ga05,Kr07,Ch19}. The  main significance from the present viewpoint is then the presence of several factors of the form (\ref{R4b}), each of which can be distinguished within the asymptotic expansion.

 The most significant extension of (\ref{6.1B})  guiding the present work is a particular $\beta$ generalisation obtained by Desrosiers and Liu \cite[Gaussian case of Th.~1.2 with $m=p$, $u = \lambda$, $s_1 = \cdots = s_{2p} =0$, and a scaling of the variables $x_l \mapsto \sqrt{2N} x_l$ in the multiple integrals corresponding to $\phi_N(\sqrt{2N}u,\dots,\sqrt{2N}u)$]{DL14},
 \begin{align}\label{asy-Gau}
            \begin{aligned}
			 \left\langle\prod_{l=1}^N |\lambda-x_l|^{2 p}\right\rangle_{\mathrm{ME}_{\beta, N}\left[e^{-\beta N x^2}\right]} 
			\sim A_{\beta, p}\left(\pi \rho^{\mathrm{W}}(\lambda)\right)^{p(2-\beta) / \beta}\left(\pi N \rho^{\mathrm{W}}(\lambda)\right)^{2 p^2 / \beta} e^{2 N p\left(\lambda^2-1 / 2-\log 2\right)},
		\end{aligned}
        \end{align}
        valid for $p$ a positive integer, where
        \begin{align}\label{A-the-const}
		A_{\beta, p}=\binom{2p}{p}\prod_{j=1}^{p}{\Gamma\left(1+{2j\over\beta}\right)\over\Gamma\left(1+{2(j+p)\over\beta}\right)} 
    \end{align}
 and $\rho^{\rm W}(\lambda)$ is the global density Wigner semi-circle law for the Gaussian $\beta$ ensemble given by \eqref{semi-cir-law} below.
 We also have from \cite[Th.~1.2]{DL14} that in the Laguerre case when the parameter $a$ is independent with $N$, we have
        \begin{multline}\label{asy-Lag}
			\left\langle\prod_{l=1}^N |\lambda-x_l|^{2 p}\right\rangle_{\mathrm{ME}_{\beta, N}\left[x^{a}e^{-2\beta N x}\right]} \sim A_{\beta, p}\left({\pi\over 2} \rho^{\mathrm{MP}}(\lambda)\right)^{p(2-\beta) / \beta}\left({\pi \over 2}N \rho^{\mathrm{MP}}(\lambda)\right)^{2 p^2 / \beta} \\
			\times(4\lambda)^{2pa/\beta}e^{2 pN \left(2\lambda-1 -2\log 2\right)},
		\end{multline}
        and in the Jacobi case when the parameters $a_1,a_2$ are independent with $N$, 
            \begin{multline}\label{asy-Jac}
			\left\langle\prod_{l=1}^N |\lambda-x_l|^{2 p}\right\rangle_{\mathrm{ME}_{\beta, N}\left[x^{a_1}(1-x)^{a_2}\right]} \sim A_{\beta, p}\left(\frac{\pi}{4}\rho^{\rm J}(\lambda)\right)^{p(2-\beta)/\beta
            } \left(\frac{\pi}{2 } N\rho^{\rm J}(\lambda)\right)^{2 p^2/\beta}\\
			\times (4 \lambda)^{-2pa_1/\beta}(4-4 \lambda)^{-2p a_2/\beta}  e^{-4 p N \ln 2} .
        \end{multline}
        In these formulas, valid for $p$ a positive integer, 
        $\rho^{\rm MP}(\lambda)\text{ and } \rho^{\rm J}(\lambda) $ are the global densities of the $\beta$  Laguerre (parameter $a$ fixed) and Jacobi (parameters $a_1$ and $a_2$ fixed) ensembles, given by  \eqref{Mar-Pas-law} and \eqref{arc-sin-law} below respectively. Furthermore, it is assumed that $\lambda$ is restricted to the support of those densities.
 
 \subsection{Aims and outline}
Comparing (\ref{6.1B}) to the asymptotic formulas (\ref{asy-Gau}), (\ref{asy-Lag}) and 
(\ref{asy-Jac}), we see that missing from the latter is informative bounds on the remainder terms. Considering the Gaussian case for example, one sees that in \cite{DL14} underpinning the derivation of (\ref{asy-Gau}) is the generalisation of (\ref{R2})
\cite{De09}
 \begin{equation}\label{2.5b}
 \Big ({1 \over i} {\sqrt{2 \over \beta}}  \Big  )^{pN/2} \Big \langle  \prod_{l=1}^p \det \Big ( i \sqrt{\beta \over 2} \nu_l\mathbb I_N - H \Big ) \Big \rangle_{{\rm ME_{\beta,N}[e^{-\lambda^2};\boldsymbol \mu}]} =
\Big  \langle  \prod_{j=1}^N \det \Big ( i \sqrt{2 \over \beta} \mu_j  \mathbb I_p - H \Big ) \Big \rangle_{{\rm ME_{{4 \over \beta},p}[e^{-\lambda^2};\boldsymbol \nu}]},   
    \end{equation} 
Here  ${\rm ME_{\beta,N}[e^{-\lambda^2};\boldsymbol \mu}]$ refers to the Gaussian $\beta$ ensemble with a source (for more on this see \cite[\S 3.5]{Fo25+} and references therein. It involves $N$ auxiliary variables $\boldsymbol \mu = (\mu_1,\dots,\mu_N)$ which in the functional form defining the probability density function involves a certain generalised hypergeometric function, based on Jack polynomials, of two sets of variables. However, if $\boldsymbol \mu =
(c,\dots,c)$, i.e.~all auxiliary variables are equal, and  $\boldsymbol \nu = \mathbf 0$, there is much simplification and (\ref{2.5b}) becomes equivalent to (\ref{R2}). Our first point is that by making such a specialisation at the beginning of the calculation, the required working of the asymptotic analysis is considerably simpler, allowing in particular for the derivation of error term bounds to supplement
(\ref{asy-Gau}), (\ref{asy-Lag}) and 
(\ref{asy-Jac}).

We state this result in Proposition 
\ref{P2.1}. In all three cases, enabling our working are duality formulas,
expressing the average over the ensemble with $N$ eigenvalues of the
 $2p$-th power of the characteristic polynomial as an average over an ensemble with $2p$ eigenvalues of the  $N$-th power of the characteristic polynomial. The duality formula we use in the Gaussian case is (\ref{R2}).
 Common to these dualities is that the right hand side can be 
 expressed as a $2p$-dimensional integral
\begin{align}\label{RN}
    R_{N,\beta}[f;\mathcal{C}]=\int_{\mathcal{C}}\dots\int_{\mathcal{C}}\prod_{j=1}^{2p}e^{Nf(u_j,\lambda)}du_j\prod_{1\leq j<k\leq 2p}|u_j-u_k|^{{\frac{4}{\beta}}}
\end{align}
times some normalization constants that can be expressed by gamma function using the Selberg integral. Here, the integration contour $\mathcal{C}$ depends on the duality formula and is the real line for the Gaussian case and the unit circle for the Laguerre and Jacobi case. Steepest descent analysis \cite{DF06,DL14,DL15} is then applied to the multiple integral $R_{N,\beta}[f;\mathcal{C}]$ to obtain its asymptotic behavior.
In \S \ref{S3}, we analyze the moments of characteristic polynomials of Gaussian $\beta$ ensemble using the duality
(\ref{R2}), which serves as a typical example for applying the steepest descent method to obtain our refinement of the results of Desoriers and Liu \cite{DL14} by providing a bound of the order of the remainder. The detailed procedures for applying steepest descent method are presented, and the asymptotic formula \eqref{asy-Gau} is subsequently derived via this line of analysis. In the subsequent sections this analysis is extended to the Laguerre and Jacobi $\beta$ ensemble cases.

 In the setting of the circular ensemble and thus for the average (\ref{R4a}), the $\beta$ generalisation of (\ref{R4b}) 
 (albeit restricted to $\beta$ rational with $\beta/2 = m/n$)
 is given by the RHS of (\ref{3.qb1}). Given the circumstance for $\beta = 2$ as discussed above, we pose the question as to (\ref{3.qb1}) being a distinguished (and thus universal) factor in the large $N$ form of the analogous power of the characteristic polynomial average for the classical Gaussian, Laguerre and Jacobi $\beta$ ensembles. With knowledge of  the results (\ref{asy-Gau}), (\ref{asy-Lag}) and  (\ref{asy-Jac}), this question then reduces to showing the equivalence of (\ref{3.qb1+}) 
 with $p$ a positive integer, and
 (\ref{A-the-const}) with $\beta$ a positive rational in reduced form. The essential identity accomplishing  this is established in Proposition \ref{G-gam-eq}. In \S \ref{S2} we also highlight structural features of the asymptotic
expansions (\ref{asy-Gau}), (\ref{asy-Lag}) and  (\ref{asy-Jac})
beyond the universal factor, specifically as they relate to Gaussian fluctuation formulas for linear statistics of $\beta$ ensemble averages in the global limit.

In the
Laguerre and Jacobi $\beta$ ensemble cases 
there are distinct asymptotic formulas, when the exponents in the weights are strictly proportional to $N$, which only appear in the main body of the paper (Propositions \ref{pro-LbE2}
and \ref{P5.2}). Writing the Laguerre exponent $a=\beta(n-N+1)/2 - 1$, and the Laguerre exponents as $a_1 = \beta (n_1 - N+1)/2 -1$, $a_1 = \beta (n_1 - N+1)/2 -1$,
where $n,n_1,n_2 \ge N$ (this decomposition is motivated by considerations in multivariate statistics
(see e.g.~\cite[Props.~3.22 and 3.6.1]{Fo10}), and then setting $n,n_1,n_2$ to be proportional to $N$ with proportionality constants greater than one, the asymptotics have previously been presented in \cite[Th.~4.3]{DL14}. Note however that except for $\beta = 2$ this setting differs from ours as we take $a,a_1,a_2$ to each be strictly proportional to $N$ (there is then no additional order one term). Moreover, the accuracy of the statement of \cite[Th.~4.3]{DL14} (for which the proof is only outlined) is in some doubt; see Remarks \ref{R-C1}.2 and \ref{R-C2}.2 below.
We make several checks on our results.
Thus, in the Appendix, the asymptotic formulas are compared with those for the global density of $\beta$ ensembles, in accordance with the inter-relationship as noted in the text below (\ref{R2}), and also, in the special case $\beta = 2$, with previously obtained results from \cite{CG21}.

\section{Refinement of the asymptotic formulas and special features}\label{S2}

\subsection{An error bound for the  Desrosiers--Liu asymptotic formulas}
By specialising the duality formulas used by Desrosiers and Liu \cite{DL14},
we are able to specify
the large $N$ asymptotics, up to a multiplicative error term
\begin{equation}\label{x1}
\left(1+\mathrm{O}\left( N^{-\mathrm{min}\{2/\beta,1\}}\right)\right).
\end{equation} 
 for averaged even powers of the  characteristic polynomials
in the Gaussian, Laguerre and Jacobi $\beta$-ensembles. The details are given in \S \ref{S3}, \ref{Lag-sec} and \ref{Jac-sec} respectively.
We state our result here, with the assumption for the latter two ensembles that the parameters in the weights are independent of $N$. In \S 4.2 we extend this to the Laguerre case with $a$ proportional to $N$, and in \S 5.2 to the Jacobi case with $a_1, a_2$ proportional to $N$.

\begin{prop}\label{P2.1}
Consider the asymptotic expansions
(\ref{asy-Gau}), (\ref{asy-Lag}) and 
(\ref{asy-Jac}) for the even moments of characteristic polynomials for the classical Gaussian, Laguerre and Jacobi 
$\beta$ ensembles. In each case a bound on the remainder associated with these expansions is obtained by inserting on the the right hand side the multiplicative factor (\ref{x1}).

\end{prop}

\subsection{Barnes $G$ function form of the constant term}
For $\beta/2$ rational, the quantity $A_{\beta,p}$ as specified by (\ref{A-the-const}) can be
rewritten to give precisely the functional form on the RHS of (\ref{3.qb1}). This
identifies the latter, first derived in the context of the circular $\beta$ ensemble
\cite{Fo92j}, as the distinguished, universal factor appearing in the asymptotic expansion of
the $2p$-th power of the characteristic polynomial for all the classical $\beta$ ensembles.

\begin{prop}\label{G-gam-eq}
    Suppose that $m,n$ and $p$ are integers, then we have
\begin{align}\label{1.18}
    \prod_{j=1}^{p}\Gamma\left(s+\frac{n}{m}j\right)=n^{-\frac{p}{2}+sp+\frac{np}{2m}(1+p)}(2\pi)^{-\frac{p(n-1)}{2}}\prod_{l=0}^{n-1}\prod_{j=1}^{m}\frac{G\left(\frac{s+l}{n}+\frac{j+p}{m}\right)}{G\left(\frac{s+l}{n}+\frac{j}{m}\right)}.
\end{align}
where $\Gamma$ and $G$ are gamma and Barnes $G$-functions respectively. Make use of the same notation as in (\ref{3.qb1}) so that  $\tilde{A}_{\beta,p} |_{\beta/2=m/n}$ is given by (\ref{3.qb1+}). Then we have
\begin{equation}\label{1.18a}
A_{\beta,p} \Big |_{\beta/2=m/n} = \tilde{A}_{\beta,p} \Big |_{\beta/2=m/n}.
\end{equation}
\end{prop}

\begin{proof}
In relation to (\ref{1.18}),
 the case when $n=1$ will be established first. Extension to general $n$ can then be carried out using the gamma function multiplication formula
\begin{equation}\label{gam-mul-for}
    \prod_{k=0}^{n-1} \Gamma\!\left(z + \frac{k}{n}\right) = (2\pi)^{\frac{n-1}{2}} \, n^{\frac{1}{2} - nz} \, \Gamma(nz).
\end{equation}
The claimed identity for $n=1$ reads
\begin{align}\label{gam-Bar-for}
    \prod_{j=1}^p \Gamma\left(s+\frac{j}{m} \right)=\prod_{j=1}^{m} \frac{G\left(s+\frac{1}{m}(j+p)\right)}{G\left(s+\frac{j}{m} \right)}.
\end{align}
    In this we first consider the case $p \leqslant m$. Some factors in the product on the RHS cancel out, and we have
$$
\prod_{j=1}^{m} \frac{G\left(s+\frac{1}{m}(j+p)\right)}{G\left(s+\frac{j}{m} \right)}=\prod_{j=1}^p \frac{G\left(s+\frac{j}{m} +1\right)}{G\left(s+\frac{j}{m}\right)}=\prod_{j=1}^p \Gamma\left(s+\frac{j}{m}\right),
$$
where we have used the functional equation $G(z+1)=\Gamma(z) G(z)$. Suppose now 
(\ref{gam-Bar-for}) 
is true for $p \leqslant km$ for some $k \in \mathbb{N}$. Then for $p \leqslant (k+1)m$, we have
\begin{multline*}
\prod_{j=1}^p \Gamma\left(s+\frac{j}{m}\right)=\prod_{j=1}^{m} \Gamma\left(s+\frac{j}{m}\right) \prod_{j=1}^{p-m} \Gamma\left(s+1+\frac{j}{m}\right) \\
=\prod_{j=1}^{m} \frac{G\left(s+1+\frac{j}{m}\right)}{G\left(s+\frac{j}{m}\right)} \cdot \prod_{j=1}^{m} \frac{G\left(s+\frac{1}{m}(j+p)\right)}{G\left(s+\frac{j}{m}\right)} 
=\prod_{j=1}^{m} \frac{G\left(s+\frac{1}{m}(j+p)\right)}{G\left(s+\frac{j}{m}\right)} .
\end{multline*}
Therefore, the formula \eqref{gam-Bar-for} holds for any integer $p$ by induction. Finally, according to \eqref{gam-mul-for} we have
\begin{align*}
    \prod_{j=1}^{p}\Gamma\left(s+\frac{n}{m}j\right)&=n^{-\frac{p}{2}+sp+\frac{np}{2m}(1+p)}(2\pi)^{-\frac{p(n-1)}{2}}\prod_{l=0}^{n-1}\prod_{j=1}^{p}\Gamma\left(\frac{s+l}{n}+\frac{j}{m}\right)\\
    &=n^{-\frac{p}{2}+sp+\frac{np}{2m}(1+p)}(2\pi)^{-\frac{p(n-1)}{2}}\prod_{l=0}^{n-1}\prod_{j=1}^{m}\frac{G\left(\frac{s+l}{n}+\frac{j+p}{m}\right)}{G\left(\frac{s+l}{n}+\frac{j}{m}\right)},
\end{align*}
where use has been made of \eqref{gam-Bar-for} in the second equality.

Consider now (\ref{1.18a}). In relation to the LHS, use of the functional equation for the gamma function gives the equivalent form to (\ref{A-the-const})
$$
A_{\beta,p} = \prod_{j=1}^p 
{\Gamma\left({2j\over\beta}\right)\over\Gamma\left({2(j+p)\over\beta}\right)}.
$$
With this as the starting point, use of (\ref{1.18}) with $s=0$, and with $s=2p/\beta$
establishes the functional form on the RHS of (\ref{3.qb1}) (which we
are denoting $\tilde{A}_{\beta,p} |_{\beta/2=m/n}$), as required.
\end{proof}


\begin{remark} ${}$ \\
1.~In the case $\beta$ rational, the functional form (\ref{3.qb1}) specifying the RHS of 
(\ref{1.18a}) no longer requires that $p$ be a
positive integer. Using this form, it is then reasonable to expect that the asymptotic formulas of Proposition \ref{G-gam-eq} are valid for all Re$(p) > - {1 \over 2}$. \\
2.~For reasons ranging from the study of Gaussian multiplicative choas, to their relation to
integrable systems, to potential theory, and the two-dimensional one component Coulomb gas,
the averaged powers of moments of the characteristic polynomial for Ginibre type
(eigenvalue support in the complex plane) random matrices has attracted attention
\cite{WW19,DS22,AKS21,SSD23,By24,BFL25,BCMS25,BY25}.
In particular, from \cite[Th.~1.1]{WW19}, for the complex Ginibre ensemble GinUE (each member is an $N \times N$ random matrix with independent standard complex Gaussian entries; see the recent monograph \cite{BF24}), scaled so that the limiting eigenvalues support is the unit disk $|\lambda|<1$ (we denote this scaled ensemble by GinUE${}^*$),
one has that for large $N$, $|z|<1$ and
Re$\,(p) > -1$,
\begin{equation}\label{G1}
\Big \langle | \det (z - G) |^{2p} \Big \rangle_{{\rm GinUE}{}^*}=
N^{{p^2 / 2} }
e^{p N (|z|^2 - 1)}
{(2 \pi)^{p / 2} \over G(1 + p)}\Big (1 + {\rm o}(1) \Big );
\end{equation}
cf.~(\ref{6.1B}). In line with our finding
(\ref{1.18a})
is that the very same factor $(2 \pi)^{p \over 2}/G(1 + p)$ appears when the ensemble
GinUE${}^*$ is changed to the ensemble TrUE${}_{n,N}$ say of
$N \times N$ sub-unitary matrices formed by
deleting $n$ rows and columns of an $(n+N) \times (n+N)$ Haar distributed unitary column. Thus,
with $N/(n+N) = \mu$ and for $p \in \mathbb N$, we have from
\cite{SSD23} that for large $N,n$
\begin{multline}\label{7.Y1a} 
   \langle | \det (z \mathbb I_N -  Z) |^{2p} \rangle_{Z \in {\rm TrUE}_{N,K}}   \\= N^{p^2/2}\mu^{Nk} 
   \bigg ( {1 - \mu \over 1 - | z|^2} \bigg )^{N p (1 - 1/\mu)} 
   \bigg ( {\sqrt{1 - \mu} \over 1 - | z|^2} \bigg )^{p^2}  
   {(2 \pi )^{p/2} \over G(1+p)} \Big ( 1 + {\rm o}(1) \Big ), \quad |z| < |\mu|^{1/2}.
   \end{multline}    
\\

\end{remark}

\subsection{Gaussian fluctuation formula structural features}
It was previously remarked that the asymptotic formula for the averaged even powers of the
characteristic polynomial for the GUE has the structure (\ref{S1a}). With $\phi_1, \phi_2$ unchanged,  inspection of (\ref{asy-Gau}) shows that for the Gaussian $\beta$ ensemble the corresponding structural formula is
\begin{equation}\label{S1b}
(\phi_2(\lambda))^{p (2 - \beta)/\beta}
(N \phi_2(\lambda))^{2p^2/\beta} e^{Np \phi_1(\lambda)} A_{\beta,p}.
\end{equation} 
Here the dependence on $p, \beta$ can be recognised as being precisely the same as
in the  asymptotic formula known for the characteristic function of 
a smooth linear statistic $\sum_{l=1}^N a(x_l)$ (i.e.~$a(x)$ is smooth on
the support of the eigenvalue density $|\lambda| \le 1$). Thus in this setting,
for (continuous) $|k|$ small enough, and for a restricted class of smooth $a(x)$ 
(real analytic on $|\lambda| < 1$) which
furthermore do not increase too fast at infinity, it is known that for large $N$ \cite{Jo98,BG11} 
\begin{multline}\label{S1c}
\log \Big \langle  e^{k \sum_{l=1}^N a(x_l)} \Big \rangle_{\mathrm{ME}_{\beta, N}\left[e^{-\beta N x^2}\right]} =
k N \int_{-1}^1 a(x) \rho^{\rm W}(x) \, dx + 
  k \Big ( {1 \over \beta} - {1 \over 2} \Big ) \int_{-\infty}^\infty \Big (  \delta(x - 1) + \delta(x+1) \\ - {1 \over \pi} {1 \over (1 - x^2)^{1/2}} \mathbbm 1_{|x|<1}\Big ) a(x) \, dx 
+ {k^2 \over \beta} \sum_{n=1}^\infty n a_n^2 + {\rm o}(1),
  \end{multline}
where $a_n :=  {1 \over \pi} \int_0^\pi a(\cos \theta) \cos n \theta \, d \theta$. The right hand side being only quadratic in $k$ up to terms which go to zero with $N$, this can be viewed as a Gaussian fluctuation formula. The
average on the LHS of (\ref{S1c}) reduces to the average in (\ref{asy-Gau}) provided
we set $k=2p$ and $a(x) = a(x;\lambda) = \log |\lambda - x|$. For $|\lambda| > 1$ this
particular linear statistic satisfies the requirements for the validity of 
(\ref{S1c}), with an application given to the large deviation asymptotics of the spectral
density given in \cite{Fo12}. However the singularity of $\log |\lambda - x|$ inside the support in the case $|\lambda| < 1$ invalidates (\ref{S1c}); specifically one can
check that the sum $\sum_{n=1}^\infty n a_n^2$ does not converge. Nonetheless, evaluation of the first two terms on the RHS of (\ref{S1c}) reproduces exactly (with $k=2p$) the two
factors in (\ref{S1b}) which are proportional to $p$ in their exponent. Returning to the diverging final term in (\ref{S1c}), use of the expansion 
$\log(2|\cos \theta - \cos \phi|) = - \sum_{n=1}^\infty {2 \over n} \cos n \theta \cos n \phi$ (see e.g.~\cite[Exercises 1.4 q.4]{Fo10}) allows one 
 to compute  $a_n = - {2 \over n} \cos n \phi$, where $\lambda = \cos \phi$. Simple manipulation then gives
   \begin{equation}\label{6.1E}
  \sum_{n=1}^\infty n a_n^2 ={1 \over 2}  \Big ( \sum_{n=1}^\infty {}^* {1 \over n} - \sum_{n=1}^\infty {\cos 2 n \phi \over n} \Big ) ={1 \over 2} \sum_{n=1}^\infty {}^* {1 \over n} +
{1 \over 2}  \log |2 \sin \phi|
  \end{equation}   
(in relation to the second equality see e.g.~\cite[Eq.~(14.95)]{Fo10}), where the asterisk indicates that a regularisation of the otherwise
divergent series is required.    Choosing the latter to be $\log N$, (\ref{6.1E}) substituted in (\ref{S1c}) is seen to precisely match the term raised to the power
of $p^2$ in (\ref{S1b}).

In the cases of Laguerre and Jacobi ensembles of Proposition \ref{G-gam-eq}, for a
smooth linear statistic an expansion formula of the form (\ref{S1c}) is still expected, but with some changes to the details. For example, in the Laguerre case, results from
\cite[Prop.~3.9 with a rescaling to account for $x \mapsto x/4$]{FRW17} give that the terms proportional to $k$ on the RHS of (\ref{S1c}) should be replaced by
 \begin{equation}\label{6.1F}
k N \int_0^1 a(x) \rho^{\rm MP}(\lambda) \, d \lambda + k \int_{-\infty}^\infty
a(x) \mu(x) \, dx,
\end{equation}
where 
\begin{equation}\label{6.1G}
\mu(x) := {1 - {2 \over \beta} \over 4} \Big ( \delta(x) - \delta(x-1) \Big ) +
{a \over \beta} \bigg ( {1 \over \pi x} \Big ( {1 \over x} - 1 \Big )^{-1/2}
\mathbbm 1_{0 < x < 1} - \delta(x) \bigg ).
\end{equation}
On the other hand, with respect to the term proportional to $k^2$, the only change is that the factor $a(\cos \theta)$ in the definition of $a_n$ is to be replaced
by $a((1+\cos \theta)/2)$; see e.g.~\cite[text below (3.37) with $c=0, d=1$]{Fo23}.
Evaluating (\ref{6.1F}) in the
case $a(x)  = \log |\lambda - x|$ with $0 < x < 1$ and substituting in (the modified
form of) the RHS of (\ref{S1c}) gives agreement with the corresponding terms in
(\ref{asy-Lag}). However in relation to the term proportional to $k^2$, the appropriate
modification of (\ref{6.1E}) can no longer precisely reproduce the term with exponent
$p^2$ in (\ref{S1c}), with the function of $\lambda$ therein not equal to
$\rho^{\rm MP}(\lambda)$, although the dependence on $N$ is correctly reproduced. 
Starting with the appropriate replacement of (\ref{6.1F}) from 
\cite[\S 4.3]{FRW17}, one can verify that these findings in relation to the
structure of the asymptotic expansion for the averaged power of
the characteristic polynomial in the Laguerre case relative to that for a smooth
linear statistic carry over to the asymptotic formula in the Jacobi case
(\ref{asy-Jac}).

\section{Gaussian case}\label{S3}
In this section, we focus on the even moments of characteristic polynomials for Gaussian $\beta$ ensemble with proper scaling, that is 
\begin{align}\label{m-cha-poly}
	\left\langle\prod_{l=1}^N |\lambda-x_l|^{2 p}\right\rangle_{\mathrm{ME}_{\beta, N}\left[e^{-\beta N x^2}\right]},\quad p\in\mathbb{N}.
\end{align}
The  duality formula (\ref{R2}) allows 
the $N$-dimensional integral in (\ref{m-cha-poly}) to be transformed to a $2p$-dimensional integral. Denote by
\begin{align}\label{Cb}
    C_{\beta,N}\left[\omega(x)\right]=\int_{-\infty}^{\infty}\dots\int_{-\infty}^{\infty}\prod_{1\leq j<k\leq N}|x_j-x_k|^\beta \prod_{j=1}^{N}\omega(x_j)dx_j
\end{align}
the partition function for $\mathrm{ME}_{{\beta},N}[\omega(x)]$.
Then it follows that
\begin{align}\label{A=CR}
	&\left\langle\prod_{l=1}^N |\lambda-x_l|^{2 p}\right\rangle_{\mathrm{ME}_{\beta, N}\left[e^{-\beta N x^2}\right]}\nonumber\\
	&\quad ={2^{-pN}\over C_{{4\over\beta},2p}\left[e^{-Nx^2}\right]}\int_{-\infty}^{\infty}\dots\int_{-\infty}^{\infty}\prod_{j=1}^{2p}e^{Nf^G(u_j,\lambda)}du_j\prod_{1\leq j<k\leq 2p}|u_j-u_k|^{{4\over\beta}}\nonumber\\
    &\quad={2^{-pN}\over C_{{4\over\beta},2p}\left[e^{-Nx^2}\right]}R_{N,\beta}[f^G;\mathbb{R}],
\end{align}
where $f^G(u,\lambda)=-u^2+\log(\sqrt{2}i\lambda-u)$ and $C_{{4\over\beta},2p}\left[e^{-Nx^2}\right]$ is evaluated by the Selberg integral
\begin{align}\label{3.6}
	C_{{4\over\beta},2p}\left[e^{-Nx^2}\right]=2^{-{2\over\beta}p(2p-1)}\pi^{p}N^{-p-{2\over\beta}p(2p-1)}\prod_{j=1}^{2p}{\Gamma(1+{2j\over\beta})\over \Gamma(1+{2\over\beta})}.
\end{align}
Before applying steepest descent analysis to obtain the asymptotic behavior of $R_{N,\beta}[f^G;\mathbb{R}]$, as the very first step, we need to modify the integration domain by certain contours such that the product of differences in the integrand is an analytic function. This is done by the following lemma from \cite[Lemma 1]{DF06}.
\begin{lemma}\label{lem1}
    Let $\{C_j\}$ be a set of non-intersecting contours such that $C_1$ is a simple contour going from $-\infty$ to $\infty$ and such that $C_j$ goes from $-\infty$ to $u_{j-1}$ for all $j = 2, \dots, n$ (see Fig.~1). Then
\begin{align*}
    \int_{-\infty}^\infty& du_1 \cdots \int_{-\infty}^\infty du_n  \prod_{i=1}^n e^{Nf(u_i,\lambda)} \prod_{1 \leq j < k \leq n} |u_j - u_k|^{4/\beta}\\
&= n! \int_{C_1} du_1 \cdots  \int_{C_n} du_n\prod_{i=1}^n e^{Nf(u_i,\lambda)} \prod_{1 \leq j < k \leq n} (u_j - u_k)^{4/\beta} ,
\end{align*}
where $-\pi < \arg u_j \leq \pi$ and where $\arg (u_i - u_j)^{4/\beta} = 0$ when $u_i, u_j \in \mathbb{R}$ but $u_i > u_j$.
\end{lemma}

\begin{figure}
    \centering
    \includegraphics[width=0.7\linewidth]{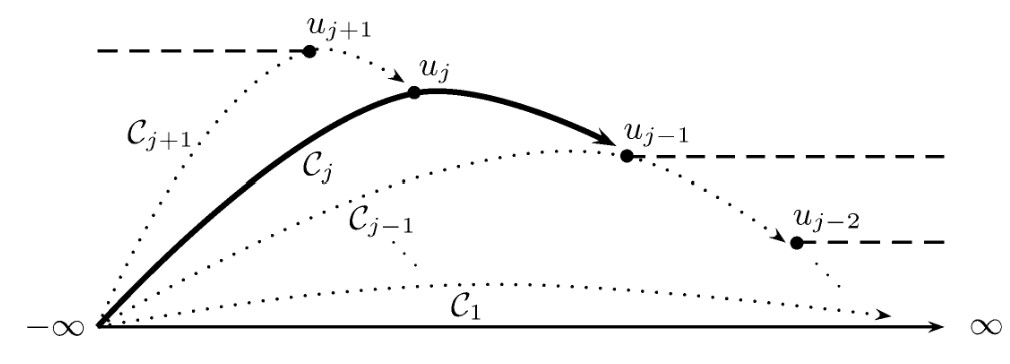}
    \caption{New contours $\{\mathcal{C}_1,\dots\mathcal{C}_n\}$ in the complex $u_j$-plane.}
    \label{fig:enter-label}
\end{figure}

The basic idea of steepest decent method is to choose a path on which $f(u,x)$ has maximum decrease. In all the classical cases we are considering, $f(u,\lambda)$ has two simple saddle points $u_{\pm}$, that is
	\begin{align}\label{f-dr-re}
		\left.\frac{\partial}{\partial u} f(u, x)\right|_{u_{ \pm}}=0,\left.\quad \frac{\partial^2}{\partial u^2} f(u, x)\right|_{u_{ \pm}}=R e^{\mathrm{i} \phi_{ \pm}}, \quad R>0 .
	\end{align}
	The directions of steepest descent at these points, denoted $\theta_{\pm}$, are such that $ \cos \left(2 \theta_{ \pm}+\phi_{ \pm}\right)=-1$ and $\sin \left(2 \theta_{ \pm}+\phi_{ \pm}\right)=0$, and therefore given by
	\begin{align}\label{theta-def}
		\theta_{ \pm}=\frac{\pi-\phi_{ \pm}}{2}(\bmod \pi), \quad-\pi<\theta_{ \pm} \leqslant \pi.
	\end{align}
\begin{prop}\label{sd-prop-G}
	Let $f (u, x)$ be a function that satisfies Eqs. \eqref{f-dr-re} and \eqref{theta-def}. Denote $f_{\pm}=f(u_{\pm},\lambda)$. Then we have
	\begin{align}\label{RNb}
		R_{N,\beta}[f;\mathbb{R}]=\binom{2p}{p}S_{p}(\Gamma_{p,\beta})^{2}\left(1+\mathrm{O}\left( N^{-\mathrm{min}\{2/\beta,1\}}\right)\right),
	\end{align}
	where 
	\begin{align}
		&S_p=\left({2\over NR}\right)^{p+{2\over\beta}p(p-1)}(u_+-u_-)^{{4\over\beta}p^2}e^{Np(f_++f_-)+i(\theta_{+}+\theta_{-})\left(p+{2\over\beta}p(p-1)\right)},\\
		&\begin{aligned}\label{Gam-nb}
			\Gamma_{n, \beta} & :=\int_{-\infty}^{\infty} d u_1 \cdots \int_{-\infty}^{\infty} d u_n \prod_{i=1}^n e^{-u_i^2} \prod_{1 \leqslant j<k \leqslant n}\left|u_j-u_k\right|^{4 / \beta} \\
			& =\frac{\pi^{n / 2}}{2^{n(n-1) / \beta}} \prod_{j=2}^n \frac{\Gamma(1+2 j / \beta)}{\Gamma(1+2 / \beta)} .
		\end{aligned}
	\end{align}
\end{prop}
\begin{proof}
The proof is similar to that used for \cite[Proposition 2]{DF06} based on steepest descent method, we write it down here for consistency. We first apply Lemma \ref{lem1} to $R_{N,\beta}[f;\mathbf{R}]$ such that the integration contours are deformed into steepest descent contours $\mathcal{S}_j=\mathcal{S}_j^+\cup \mathcal{S}_j^-$ passing through the saddle points $u_\pm$, which form a complex conjugate pair. Close to those points, the contours are parameterized by 
\begin{align*}
    u_j=u_{ \pm}+t_j e^{\mathrm{i} \theta_{ \pm}}, \quad-\pi<\arg u_j \leqslant \pi, \quad \text { on } \quad \mathcal{S}_j^{ \pm},
\end{align*}
	where the angles of steepest descent are given by \eqref{theta-def} and where $t_j\in(-\tau,\tau) $ for some $\tau>0$. Moreover, we impose $t_i>t_j$ for $i<j$ in order to guarantee $\mathfrak{R}\left(u_i\right)>\mathfrak{R}\left(u_j\right)$ for $i<j$. Let 
\begin{align*}
    y_j=\sqrt{\frac{N R}{2}} t_j,
\end{align*}
we obtain
\begin{align*}
    y_i\in (-\infty,\infty)\text{ and } Nf(u_j,\lambda)=Nf_\pm-y^2_j+\mathrm{O}(1/\sqrt{N})
\end{align*}
as $N\to \infty$. When both $u_j$ and $u_k$ are close to the same saddle point $u_{ \pm}$,
\begin{align*}
    \left(u_j-u_k\right)^{4 / \beta}=\left(\frac{2}{N R}\right)^{2 / \beta} e^{\mathrm{i} 4 \theta_{ \pm} / \beta}\left(y_j-y_k\right)^{4 / \beta}
\end{align*}
When $u_j$ is on $\mathcal{S}_j^{+}$while $u_k$ is on $\mathcal{S}_j^{-}$, we have
\begin{align*}
    \left(u_j-u_k\right)^{4 / \beta}=\left(u_{+}-u_{-}\right)^{4 / \beta}+\mathrm{O}(1 / \sqrt{N})
\end{align*}
Now in terms of the steepest descent paths, the multiple integral $R_{N,\beta}[f;\mathbb{R}]$ can be written as 
\begin{align*}
    \begin{aligned}
R_{N, \beta}[f;\mathbb{R}]= & (2p)!\sum_{n=0}^{2p}\left(\int_{\mathcal{S}_1^{+}} d u_1 \cdots \int_{\mathcal{S}_n^{+}} d u_n\right. \\
& \left.\times \int_{\mathcal{S}_{n+1}^{-}} d u_{n+1} \cdots \int_{\mathcal{S}_{2p}^{-}} d u_{2p} \prod_{j=1}^{2p} e^{N f\left(u_j,\lambda\right)} \prod_{1 \leqslant j<k \leqslant 2p}\left(u_j-u_k\right)^{4 / \beta}\right) .
\end{aligned}
\end{align*}
In terms of the new variables $y_j$, this becomes
\begin{align*}
    \begin{aligned}
& (2p)!\sum_{n=0}^{2p} S_n\left[\int_{-\infty}^{\infty} d y_1 \int_{-\infty}^{y_1} d y_2 \cdots \int_{-\infty}^{y_{n-1}} d y_n \prod_{j=1}^n e^{-y_j^2}\left(1+\mathrm{O}\left(\frac{1}{\sqrt{N}}\right)\right)\right. \\
& \quad \times \prod_{1 \leqslant p<q \leqslant n}\left(y_p-y_q\right)^{4 / \beta} \int_{-\infty}^{\infty} d y_{n+1} \int_{-\infty}^{y_{n+1}} d y_{n+2} \cdots \int_{-\infty}^{y_{2p-1}} d y_{2p} \\
& \left.\quad \times \prod_{k=n+1}^{2p} e^{-y_k^2}\left(1+\mathrm{O}\left(\frac{1}{\sqrt{N}}\right)\right) \prod_{n+1 \leqslant r<s \leqslant 2p}\left(y_r-y_s\right)^{4 / \beta}\right],
\end{aligned}
\end{align*}
where
\begin{align*}
    S_p=&(u_+-u_-)^{\frac{4}{\beta}n(2p-n)}\left(\frac{2}{NR}\right)^{\frac{2}{\beta}(n^2-2pn+2p^2-p)+p}e^{N(nf_++(2p-n)f_-)}\nonumber\\
    &\times e^{\mathrm{i}(n\theta_++(2p-n)\theta_-)}e^{\frac{2}{\beta}\mathrm{i}(n(n-1)\theta_++(2p-n)(2p-n-1)\theta_-)}.
\end{align*}
Here the multiple integrals with Gaussian weight are in fact ordered versions of $\Gamma_{n,\beta}$ introduced in \eqref{Gam-nb}. We again use Lemma \ref{lem1} and get
\begin{align*}
    R_{N,\beta}&[f;\mathbb{R}]=\sum_{n=0}^{2p}\binom{2p}{n}\left(S_n\Gamma_{n,\beta}\Gamma_{2p-n,\beta}+\mathrm{O}\left(\frac{1}{N}\right)\right)\nonumber\\
    =&\binom{2p}{p}S_p(\Gamma_{p,\beta})^2\left(1+\sum_{k=1}^{p}\frac{(p!)^2\Gamma_{p+k,\beta}\Gamma_{p-k,\beta}}{(p+k)!(p-k)!(\Gamma_{p,\beta})^2}\frac{S_{p+k}+S_{p-k}}{S_p}+\mathrm{O}\left(\frac{1}{N}\right) \right).
\end{align*}
Here the error term becomes $\mathrm{O}(1/N)$ due to the fact that the $1/\sqrt{N}$ terms are given by symmetric polynomials in $\{y_j\}$ of odd order, which equals zero after integration.
Note also that $(S_{p+k}+S_{p-k})/{S_p}$ is of order $1/N^{2k^2/\beta}$, which leads to the error term in \eqref{RNb}.
\end{proof}

In the case of Gaussian $\beta$ ensemble, and require that $|\lambda| < 1$ and thus inside the support of the limiting global density, the 
Wigner semi-circle law 
\begin{align}\label{semi-cir-law}
    \rho^{\rm W}(\lambda)={2\over\pi}\sqrt{1-\lambda^2}.
\end{align}
Solving for the saddle points of $f^G(u,\lambda)=-u^2+\log(\sqrt{2}i\lambda-u)$ yields
\begin{align*}
	u_{\pm}={i\lambda\pm \sqrt{1-\lambda^2}\over\sqrt{2}},\quad \left.\frac{\partial^2}{\partial u^2} f^G(u, \lambda)\right|_{u_{ \pm}}=4\sqrt{1-\lambda^2}\left(-\sqrt{1-\lambda^2}\mp i\lambda\right),
\end{align*}
which then gives
\begin{align*}
	R=4\sqrt{1-\lambda^2},\quad u_{+}-u_{-}=\sqrt{2(1-\lambda^2)},\quad f^G_{+}+f^G_{-}=2\lambda^2-1-\log2,\quad \theta_{ \pm}=\mp \arcsin \lambda.
\end{align*}
Therefore, by \eqref{RNb}, the large $N$ asymptotics for $R_{N,\beta}[f^G;\mathbb{R}]$ is given by
\begin{align*}
	R_{N,\beta}[f^G;\mathbb{R}]\sim\pi^{p}2^{-{2p^2\over \beta}-{2p\over \beta}(p-1)}&\left(N\pi\rho^{\rm W}(\lambda)\right)^{-p-{2p\over\beta}(p-1)}\left(\pi\rho^{\rm W}(\lambda)\right)^{4p^2\over \beta}\\
    &\times 
    e^{pN(2\lambda^2-1-\log 2)}\binom{2p}{p}\prod_{j=1}^{p}\left({\Gamma(1+2j/\beta)\over \Gamma(1+2/\beta)}\right)^{2}.
\end{align*}
By substituting this into equation \eqref{A=CR}, we obtain the sought asymptotic formula 
(\ref{asy-Gau}).

\section{Laguerre case}\label{Lag-sec}
For the Laguerre weight $x^{a}e^{-2\beta Nx}$, there are two cases of interest, and we treat them saparately.
\subsection{The case with parameter $a$ fixed.}
In this subsection, we consider the average characteristic polynomials for Laguerre $\beta$-ensemble of fixed parameter $a$.
Following a similar strategy as in the Gaussian case, the duality formula \cite[eq. (3.14)]{Fo25+}
\begin{align}\label{dualf-L}
 &\left\langle\operatorname{det}\left(\lambda \mathbb{I}_N-H\right)^{2p}\right\rangle_{\mathrm{ME}_{\beta, N}\left[x^a e^{-\beta x / 2}\right]}\nonumber\\
 &\qquad={C_{\beta, N}\left[x^{a+2 p} e^{-\frac{\beta}{2} x}\right]\over C_{\beta, N}\left[x^{a} e^{-\frac{\beta}{2} x}\right]}\left\langle\prod_{l=1}^{2p} e^{-\lambda e^{2 \pi i \theta_l}}\right\rangle_{\mathrm{CE}_{4 / \beta, 2p}\left[e^{2 \pi i(a+1) \theta / \beta-\pi i \theta(N+1)}\left|1+e^{2 \pi i \theta}\right|^{2(a+1) / \beta+N-1}\right]},
\end{align}
is applied to get
\begin{align}
    &\left\langle\prod_{l=1}^N |\lambda-x_l|^{2 p}\right\rangle_{\mathrm{ME}_{\beta, N}\left[x^{a}e^{-2\beta N x}\right]}\nonumber\\
    &\qquad =\frac{(4 N)^{-[a+1+2 p] N-\frac{\beta}{2} N(N-1)}  C_{\beta, N}\left[x^{a+2 p} e^{-\frac{\beta}{2} x}\right]}{C_{\beta, N}\left[x^a e^{-2 N \beta x}\right]}\left\langle\prod_{l=1}^{2p} e^{-4N\lambda e^{2 \pi i \theta_l}}\right\rangle_{\mathrm{CE}_{4 / \beta, 2p}\left[\omega(\theta)\right]}\nonumber\\
    &\qquad=\frac{(4 N)^{-[a+1+2 p] N-\frac{\beta}{2} N(N-1)}  C_{\beta, N}\left[x^{a+2 p} e^{-\frac{\beta}{2} x}\right]}{C_{\beta, N}\left[x^a e^{-2 N \beta x}\right]\tilde{C}_{{4\over\beta},2p}[\omega(\theta)](2\pi\mathrm{i})^{2p}}\cdot R_{N,\beta}[f^{L};\mathcal{C}], \label{eq-RN-L}
\end{align}
where $\omega(\theta)=e^{2 \pi i(a+1) \theta / \beta-\pi i \theta(N+1)}\left|1+e^{2 \pi i \theta}\right|^{2(a+1) / \beta+N-1}$ and $R_{N,\beta}[f^L;\mathcal{C}]$ is defined by \eqref{RN}, with $\mathcal{C}$ being the contour starting and ending at $z=-1$ going counterclockwise along the unit circle and 
\begin{align*}
    f^L(u, \lambda)=-\ln u+\ln (1+u)-4 \lambda u+\frac{1}{N}\left(-\ln u+\left(\frac{2}{\beta}(a+1)-1\right) \ln (1+u)\right).
\end{align*}
A slightly different lemma is used to modify the unit circle to proper contours such that the integrand in $R_{N,\beta}[f^L;\mathcal{C}]$ is analytic (See \cite[Lemma 4]{DF06}).
\begin{lemma}\label{lem2}
    Let $\left\{\mathcal{C}_j\right\}$ be a set of non-intersecting counterclockwise contours around the origin, all starting at $u_j=1$, such that $0 \leqslant \arg \left(u_n\right) \leqslant \cdots \leqslant \arg \left(u_1\right) \leqslant 2 \pi$. Then

$$
\begin{aligned}
& \oint_{\mathcal{C}} d u_1 \cdots \oint_{\mathcal{C}} d u_n \prod_{i=1}^n e^{N f\left(u_i, x\right)} \prod_{1 \leqslant j<k \leqslant n}\left|u_j-u_k\right|^{4 / \beta} \\
& \quad=n!(-1)^{n(n-1) / \beta} \int_{\mathcal{C}_1} d u_1 \cdots \int_{\mathcal{C}_n} d u_n \prod_{i=1}^n e^{N \tilde{f}\left(u_i, x\right)} \prod_{1 \leqslant j<k \leqslant n}\left(u_j-u_k\right)^{4 / \beta},
\end{aligned}
$$
where $\tilde{f}(u, x)=f(u, x)-\frac{2(n-1)}{\beta N} \ln u$.
\end{lemma}

\begin{prop}\label{sd-prop-LJ}
    Let $\tilde{f}(u, x)=f(u, x)-{2(2p-1)\over \beta N} \ln u$, where $f (u, x)$ is the function appearing in $R_{N,\beta}[f;\mathcal{C}]$ and satisfying Eqs. \eqref{f-dr-re} and \eqref{theta-def}. Let also $\tilde{f}_{\pm}=\tilde{f}(u_{\pm},x)$. Suppose moreover that the saddle points are such that $0\leq \mathrm{arg}(u_+)<\mathrm{arg}(u_-)\leq 2\pi$. Then we have
	\begin{align}\label{RNb-LJ}
		R_{N,\beta}=(-1)^{2p(2p-1)/\beta}\binom{2p}{p}S_{p}(\Gamma_{p,\beta})^{2}\left(1+\mathrm{O}\left( N^{-\mathrm{min}\{2/\beta,1\}}\right)\right),
	\end{align}
	where 
	\begin{align*}
		&S_p=\left({2\over NR}\right)^{p+{2\over\beta}p(p-1)}(u_+-u_-)^{{4\over\beta}p^2}e^{Np(\tilde{f}_++\tilde{f}_-)+i(\theta_{+}+\theta_{-})\left(p+{2\over\beta}p(p-1)\right)},\\
		&\begin{aligned}
			\Gamma_{p, \beta} =\frac{\pi^{p / 2}}{2^{p(p-1) / \beta}} \prod_{j=2}^p \frac{\Gamma(1+2 j / \beta)}{\Gamma(1+2 / \beta)} .
		\end{aligned}
	\end{align*}
\end{prop}
\begin{proof}
    We first use Lemma \ref{lem2}. The rest of the steps are similar to those of Proposition \ref{sd-prop-G}.
\end{proof}

Neglecting the $1/N$ terms in $f^L(u,\lambda)$, we solve for the saddle points of $f^{L*}(u,\lambda)=-\ln u+\ln (1+u)-4 \lambda u$, resulting in two conjugate roots
\begin{align*}
    u_{ \pm}=\frac{1}{2}\left(-1 \pm \mathrm{i} \sqrt{\frac{1}{\lambda}-1}\right) ,\quad\lambda \in(0,1).
\end{align*}
Taking second derivative gives
\begin{align*}
\begin{aligned}
& \left.\frac{\partial^2}{\partial u^2} f^L\right|_{u_{ \pm}}=\frac{1}{u^2}-\left.\frac{1}{(u+1)^2}\right|_{u_{ \pm}} =\left.\frac{2 u+1}{(u(u+1))^2}\right|_{u_{ \pm}}\\
& \qquad = \pm 16 \lambda^2 \mathrm{i}\sqrt{\frac{1}{\lambda}-1}.
\end{aligned}
\end{align*}
Therefore, we have
\begin{align*}
    R=16 \lambda^2 \sqrt{\frac{1}{\lambda}-1}, \quad \theta_{+}=-\frac{3 \pi}{4},\quad \theta_-=-\frac{\pi}{4}.
\end{align*}
Evaluation of $\tilde{f}(u,\lambda)$ on the saddle points shows
\begin{align*}
    \tilde{f}^L_{+}+\tilde{f}^L_{-}=4 \lambda+\frac{1}{N}\left(\frac{2}{\beta}(2 p-1-(a+1))+2\right) \ln (4 \lambda).
\end{align*}
Therefore, by making use of Proposition \ref{sd-prop-LJ}, we obtain that 
\begin{small}
    \begin{align*}
    &R_{N, \beta}[f^L;\mathcal{C}]\nonumber\\
    &\quad\sim \frac{{(-\pi)}^p}{2^{\frac{2 p}{\beta}(p-1)}}\left(\frac{\pi}{4} N \rho^{\mathrm{MP}}(\lambda)\right)^{-p- \frac{2p}{\beta} (p-1)} \left(\frac{\pi}{2} \rho^{\mathrm{MP}}(\lambda)\right)^{\frac{4}{\beta} p^2}(4 \lambda)^{-\frac{2 pa}{\beta}}e^{4 p N \lambda}  \binom{2p}{p} \left(\prod_{j=1}^p \frac{\Gamma(1+2 j / \beta)}{\Gamma(1+2 / \beta)}\right)^2,
\end{align*}
\end{small}
as $N\to \infty$, where 
\begin{align}\label{Mar-Pas-law}
    \rho^{\mathrm{MP}}(\lambda)={2\over \pi}\sqrt{{1\over x}-1}, \quad 0<x<1
\end{align}
is the Mar\v{c}enko-Pastur law, which is the limiting global density. According to \eqref{eq-RN-L}, we need to calculate the large $N$ asymptotics for 
\begin{align*}
\frac{(4 N)^{-(a+1+2 p) N-\frac{\beta}{2} N(N-1)}  C_{\beta, N}\left[x^{a+2 p} e^{-\frac{\beta}{2} x}\right]}{C_{\beta, N}\left[x^a e^{-2 N \beta x}\right]\tilde{C}_{{4\over\beta},2p}[\omega(\theta)](2\pi\mathrm{i})^{2p}}.
\end{align*}
This is done by the particular Selberg integral evaluation formulas 
(see e.g.~\cite[\S 4.7]{Fo10})
\begin{align}
\int_0^{\infty} d x_1 \cdots &\int_0^{\infty} d x_N \prod_{l=1}^N x_l^{\beta a / 2} e^{-\beta x_l / 2} \prod_{1 \leq j<k \leq N}\left|x_k-x_j\right|^\beta \notag\\
&=(\beta / 2)^{-N(a \beta / 2+1+(N-1) \beta / 2)} \prod_{j=0}^{N-1} \frac{\Gamma(1+(j+1) \beta / 2) \Gamma(a \beta / 2+1+j \beta / 2)}{\Gamma(1+\beta / 2)} , \label{sel-int1}\\
M_N(a, b, \lambda) & :=\int_{-1 / 2}^{1 / 2} d \theta_1 \cdots \int_{-1 / 2}^{1 / 2} d \theta_N \prod_{l=1}^N e^{\pi i \theta_l(a-b)}\left|1+e^{2 \pi i \theta_l}\right|^{a+b} \prod_{1 \leq j<k \leq N}\left|e^{2 \pi i \theta_k}-e^{2 \pi i \theta_j}\right|^{2 \lambda} \notag\\
& =\prod_{j=0}^{N-1} \frac{\Gamma(\lambda j+a+b+1) \Gamma(\lambda(j+1)+1)}{\Gamma(\lambda j+a+1) \Gamma(\lambda j+b+1) \Gamma(1+\lambda)}. \label{sel-int-2}
\end{align}
We have 
\begin{align*}
    &\frac{(4 N)^{-(a+1+2 p) N-\frac{\beta}{2} N(N-1)}  C_{\beta, N}\left[x^{a+2 p} e^{-\frac{\beta}{2} x}\right]}{C_{\beta, N}\left[x^a e^{-2 N \beta x}\right]}=\left(2\beta N\right)^{-2 p N}  \prod_{j=0}^{N-1} \frac{\Gamma\left(a+2 p+1+\frac{j \beta}{2}\right)}{\Gamma\left(a+1+\frac{j \beta}{2}\right)}\nonumber\\
    &\qquad\sim (2 \pi)^p e^{-2 p N} (4 N)^{-2 p N}  N^{2 p(N-1)+\frac{4 p}{\beta}(a+p)+\frac{2 p}{\beta}\left(\frac{\beta}{2}+1\right)}  \prod_{j=1}^{2p}\frac{1}{\Gamma\left(\frac{2}{\beta}(a+j)\right)}, 
\end{align*}
and 
\begin{align*}
    \tilde{C}_{{4\over\beta}, 2 p}&\left[e^{(\frac{2}{\beta}(a+1)-N-1) \pi i \theta} \cdot\left|1+e^{2 \pi i \theta}\right|^{\frac{2}{\beta}(a+1)+N-1}\right]=M_{2 p}\left(\frac{2}{\beta}(a+1)-1, N, \frac{2}{\beta}\right)\nonumber\\
    &=\prod_{j=0}^{2 p-1} \frac{\Gamma\left(N+\frac{2}{\beta}(a+1)+\frac{2}{\beta} j\right) \Gamma\left(1+\frac{2}{\beta}(j+1)\right)}{\Gamma\left(N+1+\frac{2}{\beta} j\right) \Gamma\left(\frac{2}{\beta}(a+1+j)\right) \Gamma\left(1+\frac{2}{\beta}\right)}\nonumber\\
    &\sim N^{2 p\left(\frac{2}{\beta}(a+1)-1\right)}\prod_{j=0}^{2 p-1} \frac{\Gamma\left(1+\frac{2}{\beta}(j+1)\right)}{\Gamma\left(\frac{2}{\beta}(a+1+j)\right) \Gamma\left(1+\frac{2}{\beta}\right)} ,
\end{align*}
as $N\to\infty$. Combining above results together, we obtain the sought asymptotic formula
(\ref{asy-Lag}).

\subsection{The case when parameter $a$ is  $\mathrm{O}(N)$.}\label{S4.2}
Set $a={\beta N\over2}\alpha$ for some $\alpha$ independent of $N$, then by the duality formula \eqref{dualf-L}, we have
\begin{align}
    &\left\langle\prod_{l=1}^N |\lambda-x_l|^{2 p}\right\rangle_{\mathrm{ME}_{\beta, N}\left[x^{{\beta N\over2}\alpha}e^{-{\beta N\
    \over 2} x}\right]}\nonumber \\
    &\qquad =\frac{N^{-2pN}  C_{\beta, N}\left[x^{{\beta N\over2}\alpha+2 p} e^{-\frac{\beta}{2} x}\right]}{C_{\beta, N}\left[x^{{\beta N\over2}\alpha} e^{-\frac{\beta}{2} x}\right]}\left\langle\prod_{l=1}^{2p} e^{-N\lambda e^{2 \pi i \theta_l}}\right\rangle_{\mathrm{CE}_{4 / \beta, 2p}\left[\omega(\theta)\right]}\nonumber \\
    &\qquad=\frac{N^{-2pN}  C_{\beta, N}\left[x^{{\beta N\over2}\alpha+2 p} e^{-\frac{\beta}{2} x}\right]}{C_{\beta, N}\left[x^{{\beta N\over2}\alpha} e^{-\frac{\beta}{2} x}\right]\tilde{C}_{{4\over\beta},2p}[\omega(\theta)](2\pi\mathrm{i})^{2p}}\cdot R_{N,\beta}[f^{L}_2;\mathcal{C}], \label{eq-RN-L2}
\end{align}
where $\omega(\theta)=e^{((\alpha -1)N+2 /\beta-1)\pi i \theta}\left|1+e^{2 \pi i \theta}\right|^{(\alpha+1)N+2 / \beta-1}$ and
\begin{align*}
    f^L_2(u,\lambda)=-\log u-\lambda u+(\alpha+1)\log(u+1)-\frac{1}{N}\left(\log u+\left(1-\frac{2}{\beta}\right)\log(1+u)\right).
\end{align*}
Solving the equation 
\begin{align*}
    {\partial \over \partial u}f^{L*}_2 =-\frac{1}{u}-\lambda+\frac{\alpha+1}{u+1}=0,
\end{align*}
we find that up to additive terms of order $1/N$, $f_2^L(u,\lambda)$ has two saddle points
\begin{align*}
    u_{ \pm}=\frac{\alpha-\lambda \pm \mathrm{i}\sqrt{(\lambda-\alpha)^2-4 \lambda}}{2 \lambda}=\frac{1}{2 \lambda}\left(\alpha-\lambda \pm \mathrm{i} \sqrt{\left(d^2-\lambda\right)\left(\lambda-c^2\right)}\right), 
\end{align*}
where $d=\sqrt{\alpha+1}+1, c=\sqrt{\alpha+1}-1$. Note that 
for $\lambda \in (c^2,d^2)$, the imaginary part of $u_\pm$ is proportional to
\begin{align}\label{L-dens2}
        \rho^{\rm L*}(\lambda)={1\over 2\pi\lambda}\sqrt{4\lambda-(\lambda-\alpha)^2},
\end{align}
which is the leading eigenvalue density of Laguerre $\beta$-ensemble with $a=\frac{\beta N}{2}\alpha$. Taking second derivative gives
\begin{align*}
    \left.\frac{\partial^2}{\partial u^2}f^{L*}_2\right|_{u_{ \pm}} & =\frac{1}{2(\alpha+1)}\left(\alpha\left(\lambda^2-(2 \alpha+4) \lambda+\alpha^2\right)\pm\left((\alpha+2) \lambda-\alpha^2\right) \sqrt{4 \lambda-(\alpha-\lambda)^2} \mathrm{i}\right).
\end{align*}
According to the notation used in \eqref{f-dr-re} and \eqref{theta-def}, we have
\begin{align*}
    R & =\frac{\lambda}{\sqrt{\alpha+1}} \sqrt{4 \lambda-(\alpha-\lambda)^2}, \quad \phi_{ \pm}= \pm\left(\pi-\arcsin \left(\sqrt{\alpha+1} \cdot\left(\alpha+2-\frac{\alpha^2}{\lambda}\right)\right)\right), \\ \theta_{+} & =\pi+\frac{1}{2} \arcsin \left(\sqrt{\alpha+1}\left(\alpha+2-\frac{\alpha^2}{\lambda}\right)\right), \theta_{-}=-\frac{1}{2} \arcsin \left(\sqrt{\alpha+1}\left(\alpha+2-\frac{\alpha^2}{\lambda}\right)\right).
\end{align*}
Further calculation shows that
\begin{align*}
    \tilde{f}^L_{2,+}+\tilde{f}^L_{2,-} =\lambda-\alpha+\left(-\alpha+\frac{1}{N}\left(2\left(1-\frac{2}{\beta}\right)\right)+\frac{4 p}{\beta}\right) \ln \lambda+\left(\alpha+1-\frac{1}{N}\left(1-\frac{2}{\beta}\right)\right) \ln (\alpha+1).
\end{align*}
Therefore, by Proposition \ref{sd-prop-LJ}, we have
\begin{align*}
    &R_{N,\beta}[f^L_2;\mathcal{C}]\sim (-2\pi )^p (\alpha+1)^{(\alpha+1)pN-{2\over p}+\frac{p}{\beta}+\frac{p^2}{\beta}}\lambda^{-\alpha pN}e^{pN(\lambda-\alpha)} \\
    &\quad \times\left(2\pi N\rho_{(1)}^L(\lambda)\right)^{-p-\frac{2p}{\beta}(p-1)}\left(2\pi \rho_{(1)}^L(\lambda)\right)^{\frac{4p^2}{\beta}}\binom{2p}{p}\left(\prod_{j=1}^{p}\frac{\Gamma(1+2j/\beta)}{\Gamma(1+2/\beta) } \right)^2.
\end{align*}
For the remaining part in \eqref{eq-RN-L2}, by making use of \eqref{sel-int1}-\eqref{sel-int-2}, we see that as $N\to \infty$, they have the following  asymptotics
\begin{align*}
    &\frac{N^{-2pN}  C_{\beta, N}\left[x^{{\beta N\over2}\alpha+2 p} e^{-\frac{\beta}{2} x}\right]}{C_{\beta, N}\left[x^{{\beta N\over2}\alpha} e^{-\frac{\beta}{2} x}\right]}\sim e^{-2pN}(\alpha+1)^{2pN}\left(\frac{\alpha+1}{\alpha}\right)^{2\alpha pN-p+\frac{4p^2}{\beta}+\frac{2p}{\beta}},\\
    &\tilde{C}_{{4\over\beta},2p}[\omega(\theta)]\sim (2\pi)^{-p}N^{-p-\frac{2p}{\beta}(2p-1)}(\alpha+1)^{2pN}\left(\frac{\alpha+1}{\alpha}\right)^{2\alpha pN-p+\frac{4p^2}{\beta}+\frac{2p}{\beta}}\prod_{j=0}^{2p-1}\frac{\Gamma\left(1+{2\over\beta}(j+1)\right)}{\Gamma\left(1+\frac{2}{\beta}\right)}.
\end{align*}

Therefore, by combining the results above, we have the companion to (\ref{asy-Lag})
when the Laguerre parameter is proportional to $N$.
	\begin{prop}\label{pro-LbE2}
	Let $p$ be a positive integer. For large $N$ we have
	\begin{small}
	    \begin{align}\label{A-LbetaE2}
		\begin{aligned}
			\left\langle\prod_{l=1}^N |\lambda-x_l|^{2 p}\right\rangle_{\mathrm{ME}_{\beta, N}\left[x^{\frac{\beta N}{2}\alpha}e^{-{\beta N\over 2} x}\right]}=&A_{\beta, p}\left({2\pi\sqrt{\alpha+1}} \rho^{\rm L*}(\lambda)\right)^{p(2-\beta) / \beta}\left({2\pi N\sqrt{\alpha+1}} \rho^{\rm L*}(\lambda)\right)^{2 p^2 / \beta} \\
			&\quad\times e^{pN \left(\lambda-\alpha\log\lambda+(\alpha+1)\log(\alpha+1)-\alpha -2\right)}\left(1+\mathrm{O}\left( N^{-\mathrm{min}\{2/\beta,1\}}\right)\right),
		\end{aligned}
	\end{align}
	\end{small}
	where $A_{\beta, p}$ is given by \eqref{A-the-const} and $\rho^{\rm L*}(\lambda)$ is the leading eigenvalue density given by \eqref{L-dens2}.
\end{prop}

\begin{remark}\label{R-C1} ${}$ \\
1.~We observe in (\ref{A-LbetaE2}) the structured form (\ref{S1b}), for particular $\phi_1, \phi_2$, as first seen in the case of the Gaussian $\beta$ ensemble. This is thus distinct
from that observed in the case of the Laguerre ensemble with fixed Laguerre parameter. This is to be expected as the analogue of (\ref{6.1G}) in the case that the Laguerre parameter is proportional to $N$ has $(1-2/\beta)$ as a common
factor; see \cite[Prop.~3.15]{FRW17}. \\
2. As mentioned in the paragraph above Section \ref{S2}, the results in \cite[Theorem~4.3]{DL14} for varying parameters fit in our setting when $\beta=2$, with the parameters related by $\alpha+1=\gamma_1$. However, compared to our result, the constant factor $\sqrt{\alpha+1}$ in the structured function $\phi_2$ and the $\alpha$ related terms $(\alpha+1)\log (\alpha+1)-\alpha $ on the exponent were not seen in their formula.

\end{remark}

\section{Jacobi case}\label{Jac-sec}
Analogous to the Laguerre case, we treat separately the case that the exponents in
the Jacobi weight are of order unity with respect to $N$, and when they are
proportional to $N$.
\subsection{The case with fixed parameter $a_1,a_2$.}
By using the duality formula \cite[eq. (3.15)]{Fo25+} for Jacobi $\beta$-ensemble, we have
\begin{align}\label{dua-for-J1}
    \begin{aligned}
& \left\langle\operatorname{det}\left(\lambda \mathbb{I}_N-H\right)^{2p}\right\rangle_{\mathrm{ME}_{\beta, N}\left[x^{a_1}(1-x)^{a_2}\right]} \\
& \quad =\frac{C_{\beta,N}[x^{a_1+2p}(1-x)^{a_2}](1-\lambda)^{2p N}}{C_{\beta,N}[x^{a_1}(1-x)^{a_2}]}\left\langle\prod_{l=1}^p\left(1-\frac{\lambda}{1-\lambda} e^{2 \pi i \theta_l}\right)^N\right\rangle_{\mathrm{CE}_{4 / \beta, p}\left[\omega(\theta)\right]}\\
&\quad =\frac{C_{\beta,N}[x^{a_1+2p}(1-x)^{a_2}](1-\lambda)^{2p N}}{C_{\beta,N}[x^{a_1}(1-x)^{a_2}]\tilde{C}_{\frac{4}{\beta},2p}[\omega(\theta)](2\pi \mathrm{i})^{2p}}R_{N,\beta}[f^J;\mathcal{C}],
\end{aligned}
\end{align}
where $\omega(\theta)=e^{2 \pi i\left(a_1-a_2\right) \theta / \beta-\pi i \theta N}\left|1+e^{2 \pi i \theta}\right|^{2\left(a_1+a_2+2\right) / \beta+N-2}$ and 
\begin{align*}
    \begin{aligned}
f^J(u,\lambda)& =-\ln u+\ln (1+u)+\ln \left(1-\frac{\lambda}{1-\lambda} u\right) \\
& +\frac{1}{N}\left(\left(\frac{2}{\beta}\left(a_1+a_2+2\right)-2\right) \ln (1+u)-\frac{2}{\beta}\left(a_2+1\right) \ln u\right).
\end{aligned}
\end{align*}
Up to order $1/N$ terms, this function has two saddle points located on
\begin{align}
    u_{ \pm}= \pm i \sqrt{\frac{1}{\lambda}-1}.
\end{align}
Since our interest is in when $0<\lambda<1$, these are a complex conjugate pair.
Moreover, we have
\begin{align*}
    \begin{aligned}
& \left.\frac{\partial^2}{\partial u^2}f^{J*}\right|_{u_\pm}=\frac{1}{u^2}-\frac{1}{(u+1)^2}-\left.\frac{\lambda^2}{(1-\lambda-\lambda u)^2}\right|_{u_{ \pm}}=-4 \lambda^2 \pm \frac{2(1-2 \lambda) \lambda^{3\over 2}}{\sqrt{1-\lambda}} \mathrm{i},
\end{aligned}
\end{align*}
which then gives
\begin{align*}
&R=\frac{2 \lambda^{\frac{3}{2}}}{\sqrt{1-\lambda}}, \quad \phi_{ \pm}= \pm \pi \mp \arcsin (1-2 \lambda),\\
&\theta_+=\frac{1}{2} \arcsin (1-2 \lambda)-\pi,\quad
\theta_{-}=-\frac{1}{2} \arcsin (1-2 \lambda).
\end{align*}
Substituting $u_\pm$ into $\tilde{f}^J(u,\lambda)=f^J(u,\lambda)-{2(2p-1)\over \beta N} \ln u$ gives
\begin{align*}
    \tilde{f}_{+}+\tilde{f}_{-}=-2 \ln (1-\lambda)-\frac{1}{N}\left(\left(\frac{2}{\beta}\left(a_1+1-2 p+1\right)-2\right) \ln \lambda+\frac{2}{\beta}\left(a_2+1+2 p-1\right) \ln (1-\lambda)\right).
\end{align*}
According to Proposition \ref{sd-prop-LJ}, we have
\begin{align*}
    R_{N,\beta}[f^J;\mathcal{C}]\sim& (-\pi)^p2^{\frac{2p^2}{\beta}+\frac{2p}{\beta
    }}N^{-p-\frac{2 p}{\beta}(p-1)}(1-\lambda)^{-2 p N} \lambda^{-\frac{2 p}{\beta} a_1}(1-\lambda)^{-\frac{2 p}{\beta} a_2}\\
    &\times\left(\frac{1}{\sqrt{\lambda(1-\lambda)}}\right)^{\left(\frac{2}{\beta}-1\right) p+\frac{2 p^2}{\beta}}\binom{2p}{p}\left(\prod_{j=1}^{p}\frac{\Gamma(1+2j/\beta)}{\Gamma(1+2/\beta)}\right)^2.
\end{align*}
The normalization constants in \eqref{dua-for-J1} in front of $R_{N,\beta}$ is evaluated by the (modified) Selberg integral \eqref{sel-int-2} and the Selberg integral in its
original form \cite[\S 4.1]{Fo10}
\begin{align}\label{sel-int3}
\begin{aligned}
    S_N\left(\lambda_1, \lambda_2, \lambda\right):=&\int_0^1 d t_1 \cdots \int_0^1 d t_N \prod_{l=1}^N t_l^{\lambda_1}\left(1-t_l\right)^{\lambda_2} \prod_{1 \leq j<k \leq N}\left|t_k-t_j\right|^{2 \lambda}\\
    =&\prod_{j=0}^{N-1} \frac{\Gamma\left(\lambda_1+1+j \lambda\right) \Gamma\left(\lambda_2+1+j \lambda\right) \Gamma(1+(j+1) \lambda)}{\Gamma\left(\lambda_1+\lambda_2+2+(N+j-1) \lambda\right) \Gamma(1+\lambda)},
\end{aligned}
\end{align}
which gives the following large $N$ asympototics
    \begin{align*}
    &\frac{C_{\beta,N}[x^{a_1+2p}(1-x)^{a_2}]}{C_{\beta,N}[x^{a_1}(1-x)^{a_2}](2\pi \mathrm{i})^{2p}}\sim  (-\pi)^{-p} 2^{-4 p N-\frac{4 p}{\beta}\left(a_1+a_2\right)-\frac{4 p^2}{\beta}+2p-\frac{6 p}{\beta}} \\
    &\qquad\qquad \times N^{-p+\frac{4 p^2}{\beta}+\frac{4 p}{\beta} a_1+\frac{2 p}{\beta}}\prod_{j=0}^{2p-1}\frac{1}{\Gamma\left(\frac{2}{\beta}(j+a_1+1)\right)},\\
    &\tilde{C}_{\frac{4}{\beta},2p}[e^{2 \pi i\left(a_1-a_2\right) \theta / \beta-\pi i \theta N}\left|1+e^{2 \pi i \theta}\right|^{2\left(a_1+a_2+2\right) / \beta+N-2}]\\
    &\qquad\qquad\sim N^{\frac{4 p}{\beta} a_1+\frac{4 p}{\beta}-2 p}\prod_{j=0}^{2p-1}\frac{\Gamma\left(1+{2\over\beta}(j+1)\right)}{\Gamma\left(\frac{2}{\beta}(j+a_1+1)\right)\Gamma\left(1+\frac{2}{\beta}\right)}.
\end{align*}

Introducing the limiting eigenvalue density \cite[Prop.~3.6.3]{Fo10}
\begin{align}\label{arc-sin-law}
         \rho^{\rm J}(\lambda)=\frac{1}{\pi\sqrt{\lambda(1-\lambda)}} \mathbbm 1_{0<x<1},
     \end{align}
and
combining these results together gives the sought asymptotic formula (\ref{asy-Jac}).

\subsection{The case with Jacobi parameters proportional to $N$.}\label{S5.2}
Let us consider the Jacobi $\beta$-ensemble with parameters $a_1=\frac{\beta N}{2}\alpha_1$ and $a_2=\frac{\beta N}{2}\alpha_2$. Using the duality formula \eqref{dua-for-J1}, we have
\begin{align*}
    & \left\langle\operatorname{det}\left(\lambda \mathbb{I}_N-H\right)^{2p}\right\rangle_{\mathrm{ME}_{\beta, N}\left[x^{\frac{\beta N}{2}\alpha_1}(1-x)^{\frac{\beta N}{2}\alpha_2}\right]} \\
& \quad =\frac{C_{\beta,N}[x^{\frac{\beta N}{2}\alpha_1+2p}(1-x)^{\frac{\beta N}{2}\alpha_2}](1-\lambda)^{2p N}}{C_{\beta,N}[x^{\frac{\beta N}{2}\alpha_1}(1-x)^{\frac{\beta N}{2}\alpha_2}]}\left\langle\prod_{l=1}^p\left(1-\frac{\lambda}{1-\lambda} e^{2 \pi i \theta_l}\right)^N\right\rangle_{\mathrm{CE}_{4 / \beta, p}\left[\omega(\theta)\right]}\\
&\quad =\frac{C_{\beta,N}[x^{\frac{\beta N}{2}\alpha_1+2p}(1-x)^{\frac{\beta N}{2}\alpha_2}](1-\lambda)^{2p N}}{C_{\beta,N}[x^{\frac{\beta N}{2}\alpha_1}(1-x)^{\frac{\beta N}{2}\alpha_2}]\tilde{C}_{\frac{4}{\beta},2p}[\omega(\theta)](2\pi \mathrm{i})^{2p}}R_{N,\beta}[f^J;\mathcal{C}],
\end{align*}
where $\omega(\theta)=e^{\pi \mathrm{i}\theta(\alpha_1-\alpha_2-1)N}|1+e^{2\pi \mathrm{i}\theta}|^{(\alpha_1+\alpha_2+1)N+\frac{4}{\beta}-2} $, and 
\begin{align*}
    f^J_2(u,\lambda)=&-(\alpha_1+1)\log u+(\alpha_1+\alpha_2+1)\log (1+u)+\log\left(1-\frac{\lambda}{1-\lambda}u\right)\\
    &+\frac{1}{N}\left(-\frac{2}{\beta}\log u+\left(\frac{4}{\beta}-2\right)\log (1+u)\right).
\end{align*}
The saddle points are given by
\begin{align*}
    u_{\pm}=&\left.\frac{1}{2\lambda(\alpha_1+1)}\right((\alpha_2-\alpha_1)\lambda+\alpha_1\\
    &\left. \pm\mathrm{i}\sqrt{((\alpha_1+\alpha_2+2)^2+\alpha_1^2-\alpha_2^2)\lambda-(\alpha_1+\alpha_2+2)^2\lambda^2-\alpha_1^2}\right).
\end{align*}
Note that the square root term coincides with the numerator in the leading eigenvalue density
\begin{align}\label{lea-eigJ*}
    \rho^{\rm J*}(\lambda) =  \frac{2+\alpha_1+\alpha_2}{2 \pi} \frac{\sqrt{\left(\lambda-c_1\right)\left(c_2-\lambda\right)}}{\lambda(1-\lambda)},
\end{align}
which is parametrized by $c_1,c_2$ satisfying
\begin{align}\label{c1c2}
    \begin{aligned}
2\left(c_1+c_2-1\right) & =\frac{2\left(\alpha_1^2-\alpha_2^2\right)}{\left(\alpha_1+\alpha_2+2\right)^2}, \\
\left(2 c_1-1\right)\left(2 c_2-1\right) & =\frac{2\left(\alpha_1^2+\alpha_2^2\right)}{\left(\alpha_1+\alpha_2+2\right)^2}-1,
\end{aligned}
\end{align}
provided it is required that $\lambda \in (c_1,c_2)$ which we do henceforth.
Further calculation shows
\begin{small}
    \begin{align}
    &\left.\frac{\partial^2}{\partial u^2}f_2^J(u,\lambda)\right|_{u_\pm}=\frac{((1+\alpha_1)(2+\alpha_1+\alpha_2)\lambda^2-2(1+\alpha_2)\lambda-\alpha_1)\Delta}{2(1+\alpha_2)(1+\alpha_2+\alpha_2)(\lambda-1)^2}\\
    &\pm\frac{((1+\alpha_2)((2+\alpha_1+\alpha_2)^2\lambda^3-(6+4\alpha_1+4\alpha_2+\alpha_1\alpha_2+\alpha_1^2)\lambda^2)+(\alpha_1^2-\alpha_1\alpha_2+2\alpha_2+2)\lambda-\alpha_1^2)\sqrt{\Delta}}{2(1+\alpha_2)(1+\alpha_2+\alpha_2)(\lambda-1)^2}\mathrm{i},
\end{align}
\end{small}
where $\Delta=((\alpha_1+\alpha_2+2)^2+\alpha_1^2-\alpha_2^2)\lambda-(\alpha_1+\alpha_2+2)^2\lambda^2-\alpha_1^2$. According to the notation used in \eqref{f-dr-re} and \eqref{theta-def}, we have
\begin{align}
    R=\frac{(1+\alpha_1)^{\frac{3}{2}}\lambda\sqrt{\Delta}}{(1+\alpha_2)^{\frac{1}{2}}(1+\alpha_1+\alpha_2)^{\frac{1}{2}}(1-\lambda)},\quad \theta_++\theta_-=\pi.
\end{align}
Recall that $\tilde{f}_2^J(u,\lambda)=f_2^J(u,\lambda)-\frac{2(2p-1)}{\beta N}\log u$. Evaluating this function at the saddle points gives
\begin{align*}
    \tilde{f}_{2,+}^J+\tilde{f}_{2,-}^J=&(1+\alpha_1+\alpha_2)\log (1+\alpha_1+\alpha_2)-(1+\alpha_1)\log(1+\alpha_1)-(1+\alpha_2)\log(1+\alpha_2)\\
    &-\alpha_1\log \lambda-(\alpha_2+2)\log(1-\lambda)+\frac{4}{\beta N}\left(\left(1-\frac{\beta}{2}\right)\log(1+\alpha_1+\alpha_2)\right.\\
    &+\left.\left(p-1+\frac{\beta}{2}\right)\log(1+\alpha_1)-p\log(1+\alpha_2)-p\log(1-\lambda)+\left(p-1+\frac{\beta}{2}\right)\log\lambda\right).
\end{align*}
Thus, from Proposition \ref{sd-prop-LJ}, we know that 
\begin{align*}
    R_{N,\beta}[f_2^J;\mathcal{C}]\sim& (-2\pi)^pN^{-p-\frac{2p}{\beta}(p-1)}\left((1+\alpha_1)(1+\alpha_2)\right)^{-\frac{p}{2}\left(\frac{2}{\beta}-1\right)-\frac{3p^2}{\beta}}(1+\alpha_1+\alpha_2)^{\frac{3p}{2}\left(\frac{2}{\beta}-1\right)+\frac{p^2}{\beta}}\\
    &\quad\times\left(\frac{\sqrt{\Delta}}{\lambda(1-\lambda)}\right)^{p\left(\frac{2}{\beta}-1\right)+\frac{2p^2}{\beta}}e^{pNg_1(\alpha_1,\alpha
    _2,\lambda)}\binom{2 p}{p} \left(\prod_{j=1}^p \frac{\Gamma(1+2 j / \beta)}{\Gamma(1+2 / \beta)}\right)^2,
\end{align*}
where
\begin{align*}
    g_1(\alpha_1,\alpha
    _2,\lambda)=&(1+\alpha_1+\alpha_2)\log (1+\alpha_1+\alpha_2)-(1+\alpha_1)\log(1+\alpha_1)\\
    &-(1+\alpha_2)\log(1+\alpha_2)-\alpha_1\log \lambda-(\alpha_2+2)\log(1-\lambda).
\end{align*}
According to the Selberg integrals \eqref{sel-int-2}, \eqref{sel-int3}, we have
\begin{align}
&  \begin{aligned}
    &\frac{C_{\beta,N}[x^{\frac{\beta N}{2}\alpha_1+2p}(1-x)^{\frac{\beta N}{2}\alpha_2}]}{C_{\beta,N}[x^{\frac{\beta N}{2}\alpha_1}(1-x)^{\frac{\beta N}{2}\alpha_2}]}\sim \left(\frac{1+\alpha_1}{\alpha_1}\right)^{2 p N \alpha_1+p\left(\frac{2}{\beta}-1\right)+\frac{4 p^2}{\beta}}\\
    &\qquad\times\left(\frac{1+\alpha_1+\alpha_2}{2+\alpha_1+\alpha_2}\right)^{2 p N\left(\alpha_1+\alpha_2+1\right)+3 p\left(\frac{2}{\beta}-1\right)+\frac{4 p^2}{\beta}} \left(\frac{1+\alpha_1}{2+\alpha_1+\alpha_2}\right)^{2 p N},
\end{aligned}
    \\
   &\begin{aligned}
        &(2\pi \mathrm{i})^{2p}\tilde{C}_{\frac{4}{\beta},2p}[e^{\pi \mathrm{i}\theta(\alpha_1-\alpha_2-1)N}|1+e^{2\pi \mathrm{i}\theta}|^{(\alpha_1+\alpha_2+1)N+\frac{4}{\beta}-2}]\\
    &\quad \sim (-2 \pi)^p N^{-p-\frac{2 p}{\beta}(2 p-1)} \alpha_1^{-2 p N \alpha_1-p\left(\frac{2}{\beta}-1\right)-\frac{4 p^2}{\beta}}\left(1+\alpha_1\right)^{-2 p N\left(1+\alpha_1\right)-p\left(\frac{2}{\beta}-1\right)-\frac{4 p^2}{\beta}}\\
&\qquad\times\left(1+\alpha_1+\alpha_2\right)^{2 p N\left(1+\alpha_1+\alpha_2\right)+3 p\left(\frac{2}{\beta}-1\right)+\frac{4 p^2}{\beta}}\prod_{j=0}^{2p-1}\frac{\Gamma\left(1+{2\over\beta}(j+1)\right)}{\Gamma\left(1+\frac{2}{\beta}\right)}.
    \end{aligned}
\end{align}
The above asymptotic formulas lead to the compannion to (\ref{asy-Jac}) in the
case that the Jacobi parameters are proportional to $N$.
\begin{prop}\label{P5.2} With $p$ a positive integer, we have that for large $N$
\begin{small}
    \begin{align}\label{J31}
\begin{aligned}
    &\left\langle\prod_{l=1}^{N}(\lambda-x_l)^{2p}\right\rangle_{\mathrm{ME}_{\beta, N}\left[x^{\alpha_1\beta N/2}(1-x)^{\alpha_2\beta N/2}\right]}=A_{\beta,p}\left(\frac{2\pi\sqrt{(1+\alpha_1)(1+\alpha_2)(1+\alpha_1+\alpha_2)}}{(2+\alpha_1+\alpha_2)^{2}}N\rho^{\rm J*}(\lambda)\right)^{2p^2/\beta}\\
&\qquad\times\left(\frac{2\pi\sqrt{(1+\alpha_1)(1+\alpha_2)}(1+\alpha_1+\alpha_2)^{\frac{3}{2}}}{(2+\alpha_1+\alpha_2)^{3}}\rho^{\rm J*}(\lambda)\right)^{p(2-\beta)/\beta}
    e^{pNg_2(\alpha_1,\alpha_2,\lambda)}
    \left(1+\mathrm{O}\left( N^{-\mathrm{min}\{2/\beta,1\}}\right)\right),
\end{aligned}
\end{align}
\end{small}
    where
\begin{align}\label{g2}
    g_2(\alpha_1,\alpha_2,\lambda)=&(1+\alpha_1+\alpha_2)\log(1+\alpha_1+\alpha_2)+(1+\alpha_1)\log(1+\alpha_1)+(1+\alpha_2)\log(1+\alpha_2) \nonumber \\
&-2(2+\alpha_1+\alpha_2)\log(2+\alpha_1+\alpha_2)-\alpha_1\log\lambda-\alpha_2\log(1-\lambda),
\end{align}
$A_{\beta, p}$ is given by \eqref{A-the-const} and $\rho^{\rm J*}(\lambda)$ is the leading eigenvalue density \eqref{lea-eigJ*}.
\end{prop}

\begin{remark}\label{R-C2} ${}$
1.~As for the Gaussian $\beta$ ensemble result
(\ref{asy-Gau}) and that for the Laguerre
$\beta$ ensemble with Laguerre parameter
proportional to $N$ 
(\ref{A-LbetaE2}), we observe in (\ref{J31})
the structured form (\ref{S1b}).
As for the
Laguerre
$\beta$ ensemble with Laguerre parameter
proportional to $N$, 
this is to be expected as the analogue of (\ref{6.1G}) in the case that the Jacobi parameters are proportional to $N$ has $(1-2/\beta)$ as a common
factor; this is a consequence of \cite[Prop.~4.8]{FRW17}. \\
2.~When $\beta=2$, this is compared to \cite[Theorem~4.3]{DL14}, with parameters related by $\alpha_1+1=\gamma_1, \alpha_2+1=\gamma_2$. The term $\frac{\sqrt{(1+\alpha_1)(1+\alpha_2)(1+\alpha_1+\alpha_2)}}{(2+\alpha_1+\alpha_2)^{2}}$ and all the $\alpha_1,\alpha_2$ related constants in $g_2(\alpha_1,\alpha_2,\lambda)$ were not included in their formula.

\end{remark}

\subsection*{Acknowledgements}
This research has been supported by the Australian Research Council 
Discovery Project grants DP210102887 and DP250102552.

\section*{Appendix}
\subsection*{Global densities}
Let us denote by $\tilde{\rho}_N(\lambda)$ the eigenvalue density of a classical $\beta$ ensemble with joint PDF proportional to \eqref{1a}, with superscripts added  to distinguish the different cases. By definition, $\rho_N(\lambda)$ is given by $N$ times the integration of the joint PDF \eqref{1a} with respect to all but one variable, and thus relates to the average power of the characteristic polynomial via
\begin{align}\label{eig-den}
    \tilde{\rho}_{N+1}(\lambda)=\frac{(N+1)C_{\beta,N}[\omega(x)]}{C_{\beta,N+1}[\omega(x)]}\omega(\lambda)\left\langle\prod_{l=1}^{N}|\lambda-x_l|^\beta\right\rangle_{\mathrm{ME}_{\beta,N}[\omega(x)]},
\end{align}
where $C_{\beta,N}[\omega(x)]$ is given by (\ref{Cb}).
A so-called global scaling can be chosen such that for $N\to\infty$, the eigenvalue density of Gaussian and Laguerre $\beta$ ensemble has compact support (in the Jacobi
case the support is $(0,1)$ without changing variables). 
It is well known (see e.g.~\cite[Ch.~1 and 3]{Fo10}) that the leading eigenvalue density upon such a global scaling is given by
\begin{align}
    &\sqrt{\beta N}\rho_{N+1}^\mathrm{G}(\sqrt{\beta N}\lambda)\sim N \rho^{\mathrm{W}}(\lambda),\quad -1<\lambda<1,\label{5.34}\\
    &2\beta N\rho_{N+1}^\mathrm{L}(2\beta N\lambda)\sim N\rho^{\mathrm{MP}}(\lambda),\quad 0<\lambda<1,\\
    &\rho_{N+1}^\mathrm{J}(\lambda)\sim N\rho^\mathrm{J}(\lambda),\quad 0<\lambda<1,
\end{align}
where $\rho^{\rm W}(\lambda),\rho^{\rm MP}(\lambda)\text{ and } \rho^{\rm J}(\lambda) $ are given by \eqref{semi-cir-law}, \eqref{Mar-Pas-law} and \eqref{arc-sin-law} respectively. For the Laguerre and Jacobi cases with $\mathrm{O}(N)$ parameters, we have
\begin{align}
    &\frac{\beta N}{2}\rho_{N+1}^{\mathrm{L}*}\left(\frac{\beta N}{2}\lambda\right)\sim N\rho^{\mathrm{L}*}(\lambda),\quad c^2<\lambda<d^2,\\
    &\rho_{N+1}^{\mathrm{J}*}(\lambda)\sim N\rho^{\mathrm{J}*}(\lambda),\quad c_1<\lambda<c_2, \label{5.38}
\end{align}
where $\rho^{\rm L*}(\lambda)$ and $\rho^{\rm J*}(\lambda)$ are given by \eqref{L-dens2} and \eqref{lea-eigJ*} respectively. 
On the other hand, our asymptotic formulas of Propositions \ref{P2.1}, \ref{pro-LbE2}
and \ref{P5.2} in the case $p=\beta/2$ ($\beta$ even) give the asymptotic form of the
average in (\ref{eig-den}). Combining this with the asymptotics of the ratio of
the normalisations (\ref{Cb}) appearing in (\ref{eig-den}) as follows from their
evaluation in terms of the corresponding Selberg integrals as given by (\ref{3.6}),
(\ref{sel-int-2}) and (\ref{sel-int3}) indeed reclaims each of (\ref{5.34})--(\ref{5.38}). We remark too that in the cases of (\ref{5.34}) and (\ref{5.38}) extended terms of the asymptotic formulas for even $\beta$ are available in the literature \cite{DF06}, the order of which agree with the bound (\ref{x1}).

\subsection*{Comparisons with known results for $\beta = 2$}
In what follows, we compare our results to the asymptotic formulas in \cite{CG21} for $\beta=2$,
which have been proved for general Re$\,(p) > - {1 \over 2}$.
\begin{itemize}
    \item \textbf{The Gaussian case.} Taking $\vec{\alpha}=2p,\vec{\beta}=0,V=2x^2,W=0$ in Theorem 1 in \cite{CG21}, we see that
    \begin{align}
        \frac{G_N(2p,0,V,0)}{G_N(0,0,V,0)}=\left\langle\prod_{l=1}^N |\lambda-x_l|^{2 p}\right\rangle_{\mathrm{ME}_{2, N}\left[e^{-2N x^2}\right]}\sim \exp(C_1N+C_2\log N+C_3),
    \end{align}
    where
    \begin{align}
        &C_1=-2p\log2-\frac{p}{\pi}\int_{-1}^{1}\frac{V(x)}{\sqrt{1-x^2}}dx+pV(\lambda)=-2p\log2-p+2p\lambda,\\
        &C_2=p^2,\quad C_3=p^2\log2\sqrt{1-\lambda^2}+\log\frac{G(1+p)^2}{G(1+2p)}.
    \end{align}
    Setting $\beta=2$ in \eqref{asy-Gau} yields exactly the same formula. 
    
    \item \textbf{The Laguerre case with fixed parameter $a$.}  Take the parameters $\vec{\alpha}=(a,2p),\vec{\alpha}_0=a$ and $\vec{\beta}=0$ in \cite[Proposition~7.2]{CG21}. In terms of our notation, this gives
    \begin{align}
        \begin{aligned}\label{Lag-fix-che}
        &\log\left\langle\prod_{l=1}^{N}|\lambda-x_l|^{2p}\right\rangle_{\mathrm{ME}_{2,N}[(x+1)^ae^{-2N(x+1)}]}\\
        &\qquad\sim 2pN(\lambda-\log 2)+p^2\log\left(N\sqrt{\frac{1-\lambda}{1+\lambda}}\right)+\log\frac{G(1+p)^2}{G(1+2p)}-pa\log 2|\lambda+1|.
        \end{aligned}
    \end{align}
    By changing variables $x_l+1\to2x_l $ in the multiple integral, one has
    \begin{align}
        \left\langle\prod_{l=1}^{N}|2\lambda-1-x_l|^{2p}\right\rangle_{\mathrm{ME}_{2,N}[(x+1)^ae^{-2N(x+1)}]}=2^{2pN}\left\langle\prod_{l=1}^{N}|\lambda-x_l|^{2p}\right\rangle_{\mathrm{ME}_{2,N}[x^ae^{-4Nx}]}.
    \end{align}
    The asymptotics of the right hand side is then given by setting $\beta=2$ in \eqref{asy-Lag}
    \begin{align}
        \left\langle\prod_{l=1}^{N}|\lambda-x_l|^{2p}\right\rangle_{\mathrm{ME}_{2,N}[x^ae^{-4Nx}]}\sim A_{2,p}\left(N\sqrt{\frac{1}{\lambda}-1}\right)^{p^2}(4\lambda)^{-pa}e^{2pN(2\lambda-1-2\log 2)},
    \end{align}
    which coincides with \eqref{Lag-fix-che} upon taking exponential and setting $\lambda\to2\lambda-1$.
    \item \textbf{The Laguerre case with $\mathrm{O}(N)$ parameter $a=\frac{\beta N}{2}\alpha$.} The leading eigenvalue density \eqref{L-dens2} suggests that this case belongs to the Gaussian type weights considered in \cite{CG21}. The asymptotic formula is then given by \cite[Theorem 1]{CG21}. To apply this theorem, an affine transformation on the variable of the weight function is required such that the equilibrium measure has the standard form. Let
    \begin{align*}
        V(x)=2\sqrt{\alpha+1}x+\alpha+2-\alpha\log(2\sqrt{\alpha+1}x+\alpha+2),\quad \psi(x)=\frac{2(\alpha+1)}{\pi(2\sqrt{\alpha+1}x+\alpha+2)}.
    \end{align*}
    Then the equilibrium measure for $V$ is supported on $[-1,1]$ with density $\psi(x)\sqrt{1-x^2}$ and we have
    \begin{align}\label{aym-b=2}
        \left\langle\prod_{l=1}^{N}(\lambda-x_l)^{2p}\right\rangle_{\mathrm{ME}_{2,N}[e^{-NV(x)}]}\sim \exp\left(C_1N+C_2\log N+C_3\right),
    \end{align}
    where
    \begin{align}\label{C1-3}
    \begin{aligned}
        &C_1=-2p\log2-\frac{p}{\pi}\int_{-1}^{1}\frac{V(x)}{\sqrt{1-x^2}}dx+pV(\lambda),\quad C_2=p^2,\\
        &C_3=p^2\log\left(\frac{\pi}{2}\psi(\lambda)\right)+p^2\log 2\sqrt{1-\lambda^2}+\log\frac{G(1+p)^2}{G(1+2p)}.
    \end{aligned}
    \end{align}
    Proceed with the evaluation of the definite integral 
    \begin{align}\label{int-1}
        \int_{-1}^{1}\frac{\log(ax+b)}{\sqrt{1-x^2}}dx=\pi\log\left(\frac{b+\sqrt{b^2-a^2}}{2}\right),
    \end{align}
    we have
    \begin{align}
&C_1=p(-2\log2+\alpha\log(\alpha+1)+2\sqrt{\alpha+1}\lambda-\alpha\log(2\sqrt{\alpha+1}\lambda+\alpha+2)),\\
        &C_2=p^2,\quad C_3=p^2\log\left(\frac{2(\alpha+1)}{2\sqrt{\alpha+1}\lambda+\alpha+2}\sqrt{1-\lambda^2}\right)+\log\frac{G(1+p)^2}{G(1+2p)}.
    \end{align}
    On the other hand, change of variables in the multiple integral shows
    \begin{align}
        \left\langle\prod_{l=1}^{N}(\lambda-x_l)^{2p}\right\rangle_{\mathrm{ME}_{2,N}[e^{-NV(x)}]}=(2\sqrt{\alpha+1})^{-2pN}\left\langle\prod_{l=1}^{N}(2\sqrt{\alpha+1}\lambda+\alpha+2-x_l)^{2p}\right\rangle_{\mathrm{ME}_{2,N}[x^{ N\alpha}e^{-Nx}]},
    \end{align}
    which upon using \eqref{A-LbetaE2} gives the same asymptotic formula.

    \item \textbf{The Jacobi case with fixed parameter $a_1,a_2$.} Denote $\vec{\alpha}_1=(a_1,2p,a_2),\vec{\alpha}_2=(a_1,a_2)$ and $t_1=\lambda$. Then according to \cite[Theorem 6.2]{CG21}, we have
    \begin{align}
        &\log\frac{J_N(\vec{\alpha}_1,0,0,0)}{J_N(\vec{\alpha}_2,0,0,0)}=\log \frac{J_N(\vec{\alpha}_1,0,0,0)}{J_N(0,0,0,0)}-\log \frac{J_N(\vec{\alpha}_2,0,0,0)}{J_N(0,0,0,0)}\\
        &\sim -(2pN+p(a_1+a_2))\log2+p^2\log N-a_1p\log(1+\lambda)-a_2p\log(1-\lambda)-\frac{p^2}{2}\log(1-\lambda^2),
    \end{align}
    as $N\to\infty$. On the other hand, this ratio of Hankel determinants relates to the  averaged characteristic polynomials by
    \begin{align}
        \log &\frac{J_N(\vec{\alpha}_1,0,0,0)}{J_N(\vec{\alpha}_2,0,0,0)}=2^{2pN}\left\langle\prod_{l=1}^{N}\left(\frac{1+\lambda}{2}-x_l\right)^{2p}\right\rangle_{\mathrm{ME}[x^{a_1}(1-x)^{a_2}]},
    \end{align}
    which after using \eqref{asy-Jac} gives the same asymptotic formula.

    \item \textbf{The Jacobi case with $\mathrm{O}(N)$ parameter $a_1=\frac{\beta N}{2}\alpha_1,a_2=\frac{\beta N}{2}\alpha_2$.} Similar to the Laguerre case with $\mathrm{O}(N)$ parameter, the global density \eqref{lea-eigJ*} exhibits square root vanishing at the end points of its support, which then belongs to the Gaussian type weight. Recall that the end points are given by $c_1,c_2$ satisfying \eqref{c1c2}. Through the transformation $\lambda\to\frac{c_1+c_2}{2}+\frac{c_2-c_1}{2}\lambda$, the global density can be written as the standard form $\psi(\lambda)\sqrt{1-\lambda^2}$, where
    \begin{align}
        \psi(\lambda)=\frac{\alpha_1+\alpha_2+2}{2\pi \left(\frac{c_1+c_2}{c_2-c_1}+\lambda\right)\left(\frac{2-c_1-c_2}{c_2-c_1}-\lambda\right)}.
    \end{align}
    In terms of the notation in \cite{CG21}, the corresponding weight function is 
    \begin{align}
        V(x)=-\left(\alpha_1\log\left(\frac{c_1+c_2}{2}+\frac{c_2-c_1}{2}x\right)+\alpha_2\log\left(1-\frac{c_1+c_2}{2}-\frac{c_2-c_1}{2}x\right)\right).
    \end{align}
    According to \cite[Theorem 1]{CG21}, the asymptotics of the average $\left\langle\prod_{l=1}^{N}(\lambda-x_l)^{2p}\right\rangle_{\mathrm{ME}_{2,N}[e^{-NV(x)}]}$ is then given by \eqref{aym-b=2} with coefficients \eqref{C1-3}. 
    Making use of \eqref{int-1},  we see that 
    \begin{align}
        -\frac{p}{\pi}\int_{-1}^{1}\frac{V(x)}{\sqrt{1-x^2}}dx=p\left(\alpha_1\log(1+\alpha_1)+\alpha_2\log (1+\alpha_2)\right.\\
        +\left.(\alpha_1+\alpha_2)(\log(\alpha_1+\alpha_2+1)-2\log(\alpha_1+\alpha_2+2))\right),
    \end{align}
    which then gives
    \begin{align}
        &\left\langle\prod_{l=1}^{N}(\lambda-x_l)^{2p}\right\rangle_{\mathrm{ME}_{2,N}[e^{-NV(x)}]}\\
        &\sim \frac{G(1+p)^2}{G(1+2p)}N^{p^2}\left(\frac{2\pi ((1+\alpha_1)(1+\alpha_2)(1+\alpha_1+\alpha_2))^{1/2}}{(2+\alpha_1+\alpha_2)^2}\rho^{\mathrm{J}*}\left(\frac{c_1+c_2}{2}+\frac{c_2-c_1}{2}\lambda\right)\right)^{p^2}\\
        &\qquad\times e^{pN\left(\alpha_1\log(1+\alpha_1)+\alpha_2\log (1+\alpha_2)+(\alpha_1+\alpha_2)(\log(\alpha_1+\alpha_2+1)-2\log(\alpha_1+\alpha_2+2))-2\log2+V(\lambda)\right)}.
    \end{align}
    On the other hand, a change of variable gives
    \begin{align}
        \left\langle\prod_{l=1}^{N}\left(\frac{c_1+c_2}{2}+\frac{c_2-c_1}{2}\lambda-x_l\right)^{2p}\right\rangle_{\mathrm{ME}_{2,N}\left[x^{\frac{\beta N}{2}\alpha_1}(1-x)^{\frac{\beta N}{2}\alpha_2}\right]}\\
        =\left(\frac{c_2-c_1}{2}\right)^{2pN}\left\langle\prod_{l=1}^{N}(\lambda-x_l)^{2p}\right\rangle_{\mathrm{ME}_{2,N}[e^{-NV(x)}]}.
    \end{align}
    By setting $\beta=2$ in \eqref{J31} and noticing that 
    \begin{align}
        \frac{c_2-c_1}{2}=\frac{2((1+\alpha_1)(1+\alpha_2)(1+\alpha_1+\alpha_2))^{1/2}}{(2+\alpha_1+\alpha_2)^2},
    \end{align}
    we get the same asymptotic formula.
\end{itemize}

\small
\providecommand{\bysame}{\leavevmode\hbox to3em{\hrulefill}\thinspace}
\providecommand{\MR}{\relax\ifhmode\unskip\space\fi MR }
\providecommand{\MRhref}[2]{%
  \href{http://www.ams.org/mathscinet-getitem?mr=#1}{#2}
}
\providecommand{\href}[2]{#2}

\end{document}